\documentclass[12pt,reqno]{amsart}
\usepackage{amsmath,amssymb,graphicx,amscd,bm}
\usepackage{amsaddr}
\usepackage{mathrsfs}
\usepackage{enumerate}
\usepackage{mathrsfs}
\usepackage{dsfont}
\usepackage{amsfonts}
\usepackage{fullpage}

\usepackage[pagewise]{lineno}

\usepackage{hyperref} 
\hypersetup{
    colorlinks=true,
    linkcolor=blue,
    filecolor=magenta,      
    urlcolor=magenta,
    }

\usepackage{caption}
\usepackage{subcaption} 

\usepackage{float} 
\usepackage{color} 

\usepackage[ruled]{algorithm2e} 
\renewcommand{\P}{\mathbb{P}}
\newcommand{\Q}{\mathbb{Q}}
\newcommand{\E}{\mathbb{E}}
\newcommand{\I}{\mathbf{1}}
\newcommand{\Y}{{\mathbf Y}}
\newcommand{\FF}{\mathbb{F}}
\newcommand{\GG}{\mathbb{G}}
\newcommand{\F}{\mathcal{F}}
\newcommand{\G}{\mathcal{G}}

\newcommand{\HH}{{\mathbb H}}
\newcommand{\RR}{{\mathbb R}}
\newcommand{\NN}{{\mathbb{N}}}

\newcommand{\C}{\mathcal{C}}
\renewcommand{\S}{\mathcal{S}}
\newcommand{\N}{\mathcal{N}}
\newcommand{\D}{\mathcal{D}}

\newcommand{\finproof}{$\square$}

\newtheorem{thm}{Theorem}[section]

\newtheorem{lem}[thm]{Lemma}
\newtheorem{prop}[thm]{Proposition}

\theoremstyle{definition}
\newtheorem{defn}[thm]{Definition}

\theoremstyle{remark}
\newtheorem*{rem}{Remark}

\theoremstyle{definition}
\newtheorem{notation}{Notation}

\theoremstyle{remark}

\numberwithin{equation}{section}

\numberwithin{equation}{section}

\begin{document}

\title{A default system with overspilling contagion}

\author[1]{Delia Coculescu*}
\thanks{*Corresponding author.}
\address[1]{Institut f\"ur Mathematik, University of Z\"{u}rich,\\Winterthurerstrasse 190, 8057 Z\"{u}rich, Switzerland\\
Department of Banking and Finance\\
Plattenstrasse 14, 8032 Z\"urich, Switzerland.}
\email{delia.coculescu@math.uzh.ch}

\author[2]{Gabriele Visentin}
\address[2]{Department of Mathematics, ETH Z\"urich,\\R\"{a}mistrasse 101, 8092 Zurich, Switzerland.}
\email{gabriele.visentin@math.ethz.ch}

\date{\today}



\begin{abstract}
%
Some dynamical contagion models for default risk have been proposed in the literature, where a system (composed of individual debtors) evolves as a Markov process conditionally on the observation of its stochastic environment, with interacting intensities.

The Markovian assumption necessitates that the environment evolves autonomously and is not influenced by  the transitions of the system. We extend this classical literature and allow a default system to have a contagious impact on its environment. With a certain probability, the transition of a debtor to the default state has an impact on the system's environment. This in turn affects the transition intensities of the other debtors inside the system. 
 
Therefore, in our framework, contagion can either be contained within the default system (i.e., direct contagion from a counterparty to another) or spill from the default system over its environment (indirect contagion).
This type of model is of interest whenever one wants to capture within a model possible impacts of the defaults of a class of debtors on the more global economy and vice versa.
 
\end{abstract}
\maketitle

\section{Motivation and aims}Default events tend to cluster in time, but this phenomena can be a manifestation of diverse causes. The literature on dynamic modelling of defaults proposed so far two major mechanisms that produce this effect. First, there is the so-called cyclical correlation, i.e., the dependence of the debtors' financial situation on some common factors. One can naturally think of some macroeconomic factors that impact the default probabilities of many debtors at a time, as for instance the level of interest rates, the prices of some commodities or the business cycle; a purely statistical approach using abstract or unobserved factors is also possible, when the aim is to fit some market data, such as credit spread observations. This type of dependence between defaults has been modelled in the standard reduced-form credit risk models with conditionally independent defaults; see for instance Duffie and Singleton \cite{DuffSing03} or Lando \cite{Lando94} for an overview.

Secondly, there is the so-called counterparty risk or direct contagion, i.e., the default of one debtor represents itself a destabilising factor impacting the default rates of surviving debtors (the counterparties of a defaulted debtor). 
 
In order to have these two mechanisms of contagion operational simultaneously it is necessary to  distinguish  within the model the default system from its environment. The role of the random environment is to carry the cyclical correlation. It is generally assumed that the common factors affecting debtors' default probabilities are stochastic processes; the random environment is their natural filtration.  The default system is simply formed by the default indicator processes of each debtor, that keeps track of the default events as they occur. 

A general approach so far (a review of the literature follows in Section \ref{sec2}) is to consider that conditionally on a given realisation of the random environment, the vector of default indicator processes of the different debtors is a time inhomogeneous Markov chain. 

While the Markovian assumption is convenient, it necessitates that the environment evolves autonomously and is not influenced by the history of the default events. Our aim here is  to introduce a new source of contagion, that we call overspilling (or indirect) contagion: the one that transmits from  the default system to its environment, subsequently having a feedback effect on the system itself. The construction, by its nature is not Markovian, the default probabilities depend not only on the current state of the default system, but also on the circumstances of the occurrence of the past defaults, more precisely  the knowledge of their impact on the environment. 

The overspilling contagion can be interpreted as the impact of default events on some economic factors that themselves are default drivers. For instance interest rates, collateral values, some commodity prices are known to both impact the solvability of debtors and be impacted under extreme circumstances by some default events. We believe that this mechanism of contagion has been mostly ignored in the previous literature not because it was deemed as unimportant, but rather due to technical reasons. We propose a tractable framework to capture the transmission of risk from defaults to factors driving defaults and back.

The paper is organised as follows. In Section \ref{sec2} we introduce the precise Markovian setting that we aim to extend, and review the existing literature. Section \ref{sec3} presents the model with overspilling contagion. Starting from a conditionally independent default system, the construction is obtained via a suitable change of the probability measure.

Section \ref{sec4} introduces and comments on the main result of the paper, that is, the survival probabilities for arbitrary sets of debtors can be obtained from a system of stochastic differential equations that can be solved recursively. Importantly, these equations are depending on the initial state of the system and the evolution of the environment, as in the Markovian setting.

In Section \ref{sec:numerical_implementation} we propose numerical implementations. In particular, we show the impact of overspilling contagion on the term structure of default probabilities in a small system of entities. We also show how to use the model for the pricing of single-name credit default swaps and $k$th-to-default swaps.

Section \ref{sec5} is dedicated to the proof of the main result. In the appendices the reader can find more details on the correspondence between our results and the Markovian setting. By considering a particular case of our setting that is Markovian,  we show that the main result can be derived from the Kolmogorov forward equations. 

\section{Default models with interacting intensities: the Markovian approach}\label{sec2}

In some default risk models, the mechanism of contagion is specified via local interactions among defaultable entities; these models are conceptually and mathematically close to models of interacting particle systems developed in statistical physics.  

We here describe a model with $n<\infty$ debtors, following Frey and Backhaus \cite{FreyBack08}; more related literature is found at the end of this section. Our aim is to introduce already the notation and framework that will be used in our extension,  while reviewing the Markovian setup. 

Let $(\Omega, \G, \FF=(\F_t)_{t\geq 0}, \P)$ be a filtered probability space satisfying the usual assumptions. The filtration $\FF$ carries the relevant information about the environment of the default system. The default system itself is modelled by a multivariate process $\Y=(Y_t(1),...,Y_t(n))_{t\geq 0}$ with state space $I:=\{0,1\}^n$, where  0 is the no default (or survival) state and 1 is the default state, so that  $(Y_t(k))_{t\geq 0}$ is the indicator process of the default of the debtor $k$. 

We denote $\N:=\{1,...,n\}$. The global information  $\GG^\mathcal N=(\G^\mathcal N_t)_{t\geq 0}$ contains both the environment and the default system:
\begin{equation}\label{int}
\G^\N_t:=\mathcal H^\N_{t^+}\text{ with }\mathcal H^\N_t:=\F_t\underset{k\in\mathcal N}{\bigvee}\sigma(Y_s(k),s\leq t).
\end{equation}
It is assumed that  conditionally on $\F_\infty$, the default process $\Y$ is a time inhomogeneous Markov chain (see Appendix \ref{AppendixDef} for a definition).

The instantaneous transition rates of $\Y$ from any state $\mathbf x\in I$ to any state $\mathbf y\in I$ ($\mathbf x\neq \mathbf y$)  at time $t$ and conditionally at $\F_\infty$ are assumed to exist and to satisfy:
\begin{equation}\label{q}
\begin{cases}
q_t( \mathbf x,\mathbf y)>0 \text{ if for some }k\in\N:  \mathbf y=\mathbf x^k \text{ and }  x(k)=0\\
q_t( \mathbf x,\mathbf y)=0 \text{ else},
\end{cases}
\end{equation}where $\mathbf x^k\in I$ is obtained from $ \mathbf x=(x(i))_{i=1,...,n}\in I$ by flipping the $k^{th}$ coordinate, $x(k)$. In other words, the transition rate is non zero only  when $\mathbf y$  can be obtained from $ \mathbf x$ by flipping a single element of $\mathbf x$ from 0 to 1.  

 For any  $k\in\N$ and $\mathbf x\in I$ with  $x(k)=0$,  $q( \mathbf x, \mathbf x^k):= (q_t( \mathbf x, \mathbf x^k)(\omega),t\geq 0)$ is  a stochastic process, and is considered $\FF$ adapted. It represents the default rate of the $k^{th}$ debtor at any time $t$, given that $\Y_t=\mathbf x$.
 
 Every component of the system (that is, debtor) $k\in\N$ has a single transition time, which is from $0$ to $1$, which is interpreted as the default time: 
\[
\tau(k):=\inf\{t\geq 0\;|\; Y_t(k)=1\},\quad k\in\N.
\]

\begin{rem}
\begin{enumerate}
\item We say that debtor $k$ is in the default state, when the process $\Y$ is in any state $\mathbf x \in I$ satisfying $x(k)=1$. 
\item  We observe that with the specification in (\ref{q}), only one default event can occur at a time and for any debtor, its default state is absorbing, in the sense that no coordinate of the process $\Y$ can be reversed from 1 to 0. We shall keep these features in our extension.
\end{enumerate}
\end{rem}

\begin{defn}
We call the default intensity of  debtor $k$  (or alternatively the intensity of $\tau(k)$) with respect to $(\GG^\N,\P)$ the nonnegative process $\lambda^\N(k)$ such that 
$$\left(Y_t(k)-\int_0^{t} \lambda^\N_s(k) ds\right)$$
 is a $(\GG^\N,\P)$-martingale, whenever such a process exists.  We denote $\bm\lambda^\N:=(\lambda_t^\N(1) ,\cdots, \lambda_t^\N(n))_{t\geq 0}$ the vector of default intensities.
\end{defn}
The intensity of one debtor depends implicitly on the set of contagious debtors $\N$. More exactly, given the above transition rates, it can be shown that
 \begin{equation}\label{l0}
 \lambda_t^\N(k)=q_t(\Y_t, \Y_t^k).
 \end{equation}
In this context, most existing models either directly assume, or are consistent with,  the following  representation of the intensities: there exist stochastic processes $\bm{\lambda}:=(\lambda_t(k), k\in\N)_{t\geq 0}$ and $\bm \xi:=(\xi_t(i,j)\; i,j\in\N)_{t\geq 0}$, all being $\FF$ adapted and such that 
 \begin{equation}\label{l1}
 \lambda^\N_t(k)=\lambda_t(k) + \sum_{i\in\N} \xi_t(k,i)Y_{t}(i).
 \end{equation}
Each debtor having a single transition time, we can rewrite (\ref{l1}) as:
\begin{equation}\label{l2}
 \lambda^\N_t(k)=\lambda_t(k) + \sum_{i\in\N} \xi_t(k,i)\I_{\{\tau(i)\leq t\}}.
 \end{equation}
Default intensities as in (\ref{l2}) will arise as a special case in the construction that we propose in the next section.

 Early default models with interacting intensities are  Kusuoka \cite{Kusuok}, Davis and Lo \cite{DaviLo01}, Jarrow and Yu \cite{JarrYu01}, Yu \cite{Yu07}, Bielecki and Rutkowski \cite{BielRutk03}; in more recent years, we mention for instance Frey and Backhaus \cite{FreyBack08}, \cite{FreyBack10}, Herbertsson  \cite{Herb08}, Herbertsson and Rootz\'en \cite{HerbRoot08}, Jian and Zen \cite{JianZhen09}, Bielecki, Cr\'epey and Jeanblanc \cite{BielCrepJean10}, Bo and Capponi \cite{BoCapp16} to name only a few.
Other models of contagion are variants of the above described framework. We mention some of the variants: non absorbing default states (Giesecke and Weber \cite{GiesWebe04}, \cite{GiesWebe06}); credit migration models with more than two states for each debtor (Davis and Esparragoza-Rodriguez \cite{DaviEspa07}, Egloff et al. \cite{EgloLeipVani04},  Bielecki et al. \cite{BielCrepJeanRutk07}, Horst \cite{Horst07}); the so-called frailty models where the filtration $\FF$ is (partially) unavailable for pricing and filtering techniques are used (Frey and Schmidt \cite{FreySchmidt}, Duffie et al. \cite{DuffEckHoreSait09}); more than one default is allowed to occur at a time (Bielecki et al. \cite{BielCousCrepHerb14}).  We recommend the survey paper by  Bielecki, Cr\'epey and  Herbertsson \cite{HerbBielCerp11} for a more detailed presentations of  the Markovian setting.

All models with counterparty risk/direct contagion that fit in the mathematical framework developed above, are part of the so-called bottom-up approach within the intensity-based models. We do not mention models within the so-called top-down approach, nor models that use copulas to describe dependence, as they are not linked to the current approach.

While introducing a generalisation of the above Markovian framework in the next section, we will be close in the spirit to another part of the default risk literature, namely the enlargements of a filtration approach initiated by Elliott et al. \cite{elliotjeanbyor}. Our approach needs to allow default events to occur simultaneously with events in the filtration $\FF$; this property, even though not standard in credit risk,  has been explored in some recent papers.  We mention the following: Coculescu \cite{coculescu},  Aksamit et al. \cite{AksaChouJean14} investigated the mathematical implications of this property; while a financial application is given in Jiao and Li \cite{JiaoLi16}. The existing models deal with the case of single defaults and hence are not  studying contagion among debtors. 

The paper by Frey and Runggaldier \cite{FreyRung10} deserves a particular attention. They present a default model, where, as in our approach, default events can occur simultaneously with some external events, which are determined by a factor process $\mathbf X$.  The couple $(\mathbf X,\mathbf Y)$ is considered to be a Markov process. The process $\mathbf X$ can be considered to belong to a subfiltration  $\FF$, even though the model is not explicitly built that way, and in this case it could be interpreted as the environment of the default system, within our framework. But their model involves unobserved factors and the focus is to develop the filtering techniques that are appropriate in this particular framework. Accordingly, here is no use of the enlargement of a filtration but rather the opposite, that is, projecting on subfiltrations.  Another difference lies in the fact that we do not need a specific factor process $\mathbf X$ and its exact dynamics for deriving our results.  As far as we know,  \cite{FreyRung10} is the only preexisting paper to model explicitly the vector process $\mathbf Y$ as having direct contagion between its components and at the same time  indirect contagion driven by an exterior  process $\mathbf X$ (that they call factor process), itself depending on $\mathbf Y$.

Let us also point out the paper of  El Karoui et al. \cite{ElKaJiaoJean17}, which analyses the effects of changes of a probability measure for a default system. Their framework  is very general and flexible to encompass many possible concrete applications: the default times do not necessarily admit an intensity, they can be either ordered or not ordered, finally it accommodates many possible information sets (i.e., observations of the default system). On the opposite, our objective in this paper is very applied: we propose a specific example of a default system that "contaminates" its environment which is a generalisation of the Markovian model presented above; being specific, we are able to characterise the corresponding survival probabilities.

\section{Interacting intensities and overspilling contagion}\label{sec3}

As in the previous section, we consider a group $\N=\{1,...,n\}$ of debtors. We shall introduce the dependence structure within the group $\mathcal N$ in two steps, as follows. To begin with, we build the model under a measure $\P^0$ where the default events are independent conditionally on $\FF$, that is, we have cyclical correlation but no contagion. The channels for the transmission of the contagion from the default system to its environment are already present, but inactive under $\P^0$; they are materialised in a sequence of $\FF$ stopping times $T(k)_{k\geq 0}$, where default events can occur with positive probability.
We then shape the wished contagion (direct and indirect) via a change of the probability measure.

Consider a set $\S\subset \N$ of all debtors that are systemically contagious, i.e., their default can produce a direct or an indirect contagion. The contagion mechanism that we propose is generating default intensities of the following form:
\begin{align}\label{lambdaex}
\lambda^\mathcal S_t(i)&=\lambda_t(i)+\sum_{j\in \mathcal S}\xi^{X(j)}_t(i,j)\I_{\{\tau(j)< t\}}\quad\text{for $i\in\N$,}
\end{align}
where $X(j)\in\{A,B\}$ is a random variable, $X(j)=A$ if the default $j$ is producing a direct contagion, while $X(j)=B$ will indicate that we have indirect contagion.  The quantities $\xi^A_t(i,j)$ and $\xi^B_t(i,j)$ are a priori different quantities, but more  importantly, when $X(j)=B$ some changes are occurring in the environment, i.e., some $\FF$ adapted processes are impacted at the default event $\tau(j)$. The fact that the intensity of a surviving debtor $i$ is augmented by $\xi^B(i,j)$ as shown in (\ref{lambdaex}) is in fact a consequence of the modification of the environment. No impact on the environment occurs in the alternative case where $X(j)=A$.

We see that in such a framework, the environment does not evolve autonomously from the default system, which is precisely our objective. 

\begin{rem}
In a Markovian model as the one in the previous section, the vector of intensity processes encodes the necessary and sufficient information about the distribution of the default process $\Y$ conditionally on $\FF$ and given $\Y_0$ (the $\FF$ conditional transition rates can be obtained from $\bm \lambda^\N$ and vice-versa). For this reason, these models are also called "intensity based".
This is not the case in our framework, where we need to rely on the so-called hazard processes; a given intensity process can arise from different hazard processes (as explained in \cite{coculescunikeghbali}).  For this reason we do not provide immediately more details on the processes in (\ref{lambdaex}), that we consider to be by-products of the model. 
\end{rem}

\subsection{The model under \texorpdfstring{$\P^0$}{P0}: conditional independence}

We begin with a filtered probability space $(\Omega,\G,\FF=(\F_t)_{t\geq 0},\P^0)$ and an $\FF$-adapted and increasing process $\mathbf \Gamma=(\Gamma (k), k\in\N)$, with $\mathbf \Gamma_0=(0,..,0)$ a.s. and $\lim_{t\to \infty} \Gamma_t(k)=+\infty$ a.s., for all $k\in\N$. 

We assume the probability space supports a sequence of random variables $e(k)$, $k\in\N$ which are i.i.d. with exponential distribution with parameter $1$, and which are independent of $\F_\infty$. We define:
\[
\tau(k)=\inf\left \{t\geq 0; \; \Gamma_t(k) \geq e(k)\right \},\;k\in\N. 
\]The process $\mathbf \Gamma$ is known as the hazard process in the credit risk literature (see \cite{elliotjeanbyor}, \cite{jenbrutk1}, \cite{coculescunikeghbali}); it synthesises all the necessary information about the default time; the compensator  process of the default time can be computed starting from the hazard process, as we shall see in a moment.

In this paper, we work under the following assumptions:

\bigskip

\textbf{Assumptions.} For all $k\in\N$, there exist $\FF$-predictable processes, $\alpha(k)$, $\gamma(k)$ and $\eta(k)$, that are  nonnnegative  and bounded and such that:
\begin{itemize}
\item[A1.] The hazard process of the default time $\tau(k)$ has the representation:
\[
\Gamma_t (k)=\int_0^t \alpha_s(k)ds +\eta_{T(k)}(k)\I_{\{T(k)\leq t\}},
\]where $T(k)$ is an $\FF$ stopping time;

\item[A2.]  The $\FF$ stopping time  $T(k)$ is totally inaccessible, with intensity process $(\gamma_t(k))_{t\geq 0}$. We define the $(\FF,\P^0)$ martingales:
\begin{equation}
\label{eq:n}
n_t(k):=\I_{\{T(k)\leq t\}}-\int_0^{t\wedge T(k)}\gamma _s(k)ds.
\end{equation}
We assume that the martingales $n(k)$ and $n(j)$ are orthogonal for any $k,j\in\N$ with $k\neq j$.

 
 \end{itemize}

These assumptions permit to have a simple model, where the default times admit an intensity. The more general framework appears in Coculescu \cite{coculescu}, where only the case of a single debtor is treated. 
We point out that this model is a generalization of the so-called Cox process. Indeed, by taking $\eta \equiv 0$, we obtain a Cox process. The impact of varying $\eta$ on the survival probability of a single entity is shown in Figure \ref{fig:conditional_independence_setting}, for a small selection of (deterministic) values.

Our construction under the measure $\P^0$ leads to $\FF$ conditional survival probabilities (also known as the Az\'ema's supermartingales) that have simple forms, and where the intensities are gives as follows:

\begin{lem}
The Az\'ema's supermartingale $Z(k)=(Z_t(k),t\geq 0)$, defined as:
\[
Z_t(k):=\P^0(\tau(k)>t|\F_t)=e^{-\Gamma_t(k)}
\] 
 has a multiplicative decomposition with respect to $(\FF,\P^0)$ (i.e., in the form of a $(\FF,\P^0)$-local martingale times a decreasing and $\FF$-predictable process) given by: 
\begin{equation}\label{Z(k)}
Z_t(k)=\mathcal E_t(\nu (k))e^{-\Lambda_t (k)},
\end{equation}where:
\begin{align}  \label{nuk}
\nu_t(k):&=-\int_0^tg_s(k)dn_s(k)\\
g_t(k): &=(1-e^{-\eta_t(k)})\I_{\{T(k)\geq t\}} \label{eq:g} \\
\Lambda_t(k):& =\int_0^t \lambda_s(k)ds
\end{align}
with $\lambda(k)$, the intensity of $\tau(k)$, being:

\begin{equation}
\label{eq:lambda}
\lambda_t(k):=\alpha_t(k)+g_t(k)\gamma_t(k).    
\end{equation}

\end{lem}
 \proof
 The expression (\ref{Z(k)}) is trivial.  The fact that the process $\lambda(k)$ is precisely the intensity of $\tau(k)$ follows from a result by Jeulin and Yor (1978) that we recall in the Appendix (Theorem \ref{calccomp}). 
 \finproof
 
From the expression \eqref{eq:lambda}, we see that whenever  $\eta(k)$ strictly positive, the intensity of a default time $\tau(k)$ is higher then $\alpha(k)$ before the arrival of the stopping time $T(k)$, then it drops to the level of $\alpha(k)$. With our construction, the default can actually occur with positive probability at the time $T(k)$,  as will be shown below. Hence $T(k)$ may be interpreted as  a time of economic shock, that may trigger the default of debtor $k$. However, under $\P^0$, we notice that the intensities of different debtors do not interact, that is, an intensity $\lambda(k)$ is not affected by  a default time, say  $\tau(i)$ $i\neq k$.  Hence we have non-contagious defaults under $\P^0$.

\begin{figure}[H]
    \includegraphics[width=0.8\textwidth]{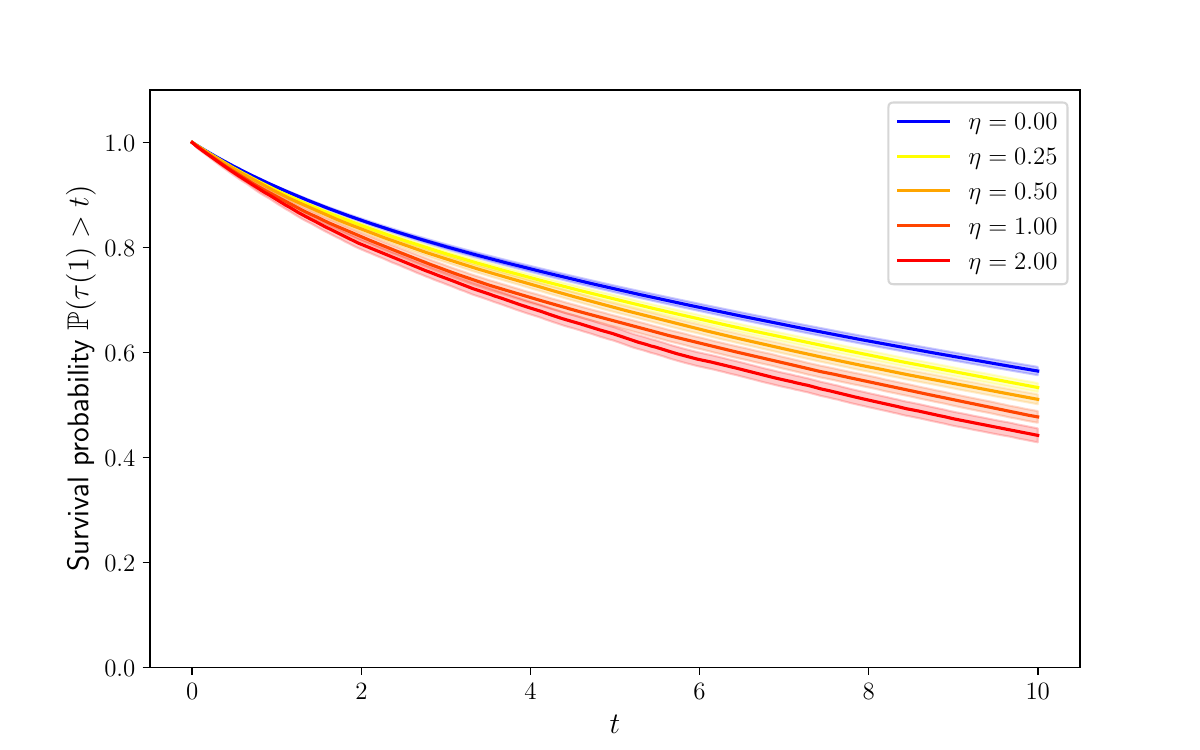}
    \caption{\footnotesize Monte Carlo estimate of survival probability $\P(\tau(1) > t)$ for a single entity with varying values of $\eta$ (10'000 MC samples, asymptotic 99.99\% confidence band). Factor process and other parameters chosen as in Section \ref{sec:numerical_implementation}.}
    \label{fig:conditional_independence_setting}    
\end{figure}
 
In order to  be able later to discriminate between direct resp. indirect contagions, we need to decompose a default time $\tau(k)$ in its ''specific" ($\tau^A(k)$) resp "systematic" ($\tau^B(k)$) counterparts, as follows:
\begin{prop}Let us consider a set $\mathcal C\subset \mathcal N$  and fix some $k\in\mathcal C$.  We define the $\GG^\mathcal C$ stopping times $\tau^A(k)$ and $\tau^B(k)$:
\begin{align*}
\tau^A(k):&=\tau(k)\I_{\{\tau(k)\neq T(k)\}}+\infty \I_{\{\tau(k)=  T(k)\}} \\
\tau^B(k):&=\tau(k)\I_{\{\tau(k) = T(k)\}}+\infty\I_{\{\tau(k)\neq T(k)\}},
\end{align*}
 so that:
\[
\tau(k)=\tau^A(k)\wedge\tau^B(k).
\]
Then,  the compensators of both $\tau^A$ and $\tau^B$ admit $(\GG^\mathcal C, \P^0)$ intensities on $\RR_+$. These are given as follows.
\begin{itemize}  
\item[(i)] For $\tau^A(k)$  the $(\GG^\mathcal C, \P^0)$-compensator is $(\int_0^{t\wedge\tau(k)}\alpha_s(k)ds)$, i.e., the intensity is $\I_{\{\tau(k)>t\}}\alpha_t(k)$.
\item[(ii)] For $\tau^B(k)$ the $(\GG^\mathcal C, \P^0)$-compensator is $(\int_0^{t\wedge\tau(k)}\beta_s(k)ds)$, i.e., the intensity is $\I_{\{\tau(k)>t\}}\beta_t(k)$, with:
\begin{align*}
\beta_t(k):&=g_t(k)\gamma_t(k).
\end{align*}
\end{itemize} \end{prop}
\proof
Let us denote by $\Lambda^A(k)$ the $(\GG^\mathcal C, \P^0)$-compensator of $\tau^A(k)$ and by $\Lambda^B(k)$ the $(\GG^\mathcal C, \P^0)$-compensator of $\tau^B(k)$.  

We compute first $\Lambda^B(k)$, which is defined as the unique increasing and $\GG^\mathcal C$-predictable process, such that for all bounded and $\GG^\mathcal C$-predictable process $H$ and for all $t\geq 0$, the following holds :
\[
\E^0\left[\int_0^tH_sd\I_{\{\tau^B(k)\leq s\}}\right]=\E^0\left[\int_0^tH_sd\Lambda^B_s(k)\right]
\]
For (any) $H$ as above (i.e.,  $\GG^\mathcal C$-predictable), there exists a $\GG^{\C-k}$-predictable process, that we denote $h$, such that: $H_t\I_{\{\tau(k)\geq t\}}=h_t\I_{\{\tau(k)\geq  t\}}$, in particular, $H_{\tau(k)}=h_{\tau(k)}$ (see for instance \cite{yorjeulin}, Lemme 1). Because $\{\tau^B(k)\leq t\}\subset \{\tau(k)\leq t\}$, we also have $H_{\tau^B(k)}\I_{\{\tau^B(k)\leq t\}}=h_{\tau^B(k)}\I_{\{\tau^B(k)\leq t\}}$ and therefore:
\begin{align*}
\E^0\left[\int_0^tH_sd\I_{\{\tau^B(k)\leq s\}}\right]&=\E^0\left[\int_0^th_sd\I_{\{\tau^B(k)\leq s\}}\right]=\E^0\left[ h_{T(k)}\I_{\{\tau(k)=T(k)\}}\I_{\{T(k)\leq t\}}\right]\\
&=\E^0\left[h_{T(k)}\P^0\left (\tau(k)=T(k)|\G^{\C-k}_{T(k)}\right) \I_{\{T(k)\leq t\}}\right].
\end{align*}
Because the random variables $e(k)$, $k\in\N$ which are independent, we find that 
\begin{align}\label{pTk}
\P^0\left (\tau(k)=T(k)|\G^{\C-k}_{T(k)}\right)&= \P^0\left (\tau(k)=T(k)|\F_{T(k)}\right)\\\nonumber
&= -\Delta Z_{T(k)}(k)=-e^{-\Gamma_{T(k)}(k)}+e^{-\Gamma_{T(k)^-}(k)}\\\nonumber
&=e^{-\int_0^{t \wedge T(k)} \alpha_s(k) ds}(1-e^{-\eta_{T(k)}(k)})\\\nonumber
&=p_{T(k)}(k),
\end{align}
where we use the notation: 

\begin{equation}
\label{eq:p}
p_t(k):=e^{-\int_0^{t \wedge T(k)} \alpha_s(k) ds}(1-e^{-\eta_{t\wedge T(k)}(k)}).  
\end{equation}

The processes $p(k)$ will play a key role also later on, and they are stopped at  $T(k)$, so that we have the simple relation:
\[
\P^0\left (\tau(k)=T(k)|\F_t\right)=\E^0[ p_{\infty}(k)|\F_t].
\]

Hence we obtain:
\begin{align*}
\E^0\left[\int_0^tH_sd\I_{\{\tau^B(k)\leq s\}}\right]&=\E^0\left[ h_{T(k)}p_{T(k)}\I_{\{T(k)\leq t\}}\right]\\
&=\E^0\left[\int_0^t h_s p_{s} (k)d\I_{\{T(k)\leq s\}}\right]\\
&= \E^0\left[\int_0^t h_{s}\frac{\I_{\{\tau(k)\geq s\}}}{Z_{s-}(k)}p_{s}(k)d\I_{\{T(k)\leq s\}}\right]\\
&= \E^0\left[\int_0^t H_s\I_{\{\tau(k)\geq s\}}g_{s}(k) d\I_{\{T(k)\leq s\}}\right]\\
&= \E^0\left[\int_0^{t\wedge T(k)} H_s\I_{\{\tau(k)\geq s\}}g_{s}(k)\gamma_s(k)ds\right].
\end{align*}
We have used the property $\frac{\I_{\{\tau(k)\geq s\}}}{Z_{s-}(k)}=e^{\Lambda_t(k)}\I_{\{\tau(k)\geq s\}}$, and, in the last step, we have used the fact that the compensator of $T(k)$ is given by $(\int_0^{t\wedge T(k)}\gamma(s)ds)$. We conclude that:
\[
\Lambda^B(k)= \int_0^t \I_{\{\tau(k)\geq s\}\cap\{T(k)\geq s\}}g_{s}(k)\gamma_s(k)ds= \int_0^{t\wedge \tau(k)} \beta_s(k)ds.
\]
Because $\I_{\{\tau(k)\leq t\}}=\I_{\{\tau^A(k)\leq t\}}+\I_{\{\tau^B(k)\leq t\}}$, we have $\Lambda^A(k)=\Lambda(k)-\Lambda^B(k)$, hence the result.
\finproof

\bigskip

Before proceeding to the next step and introducing contagion, it is useful to have a look at the survival probabilities under conditional independence, as seen from time 0. The aim is  to emphasise that a class of probability measures is handy to use.
Under $\P^0$, the time $t$ survival probability in a group $\C\subset\N$ is given by:
\begin{align}\nonumber
\P^0(\tau (k)>t,\; \forall k\in\C)& =\E^0 \left[\prod_{k\in\C}Z_t(k)\right]=\E^0 \left[\exp \left (- \sum_{k\in \C} \int_0^t\lambda_s(k)ds \right )\prod_{k\in\C}\mathcal E_t(\nu(k))\right]\\\label{p0}
&= \bar \E_{\C} \left[ \exp \left (-\sum_{k\in \C} \int_0^t\lambda_s(k)ds  \right )\right]\\
&= \bar \E_{\C} \left[ \exp \left (-\sum_{k\in \C}\Lambda_t(k) \right )\right],
\end{align}
with $\bar\E_\C$ being the expectation operator under the measure $\bar\P_\C$ defined below.

\begin{defn}\label{PC} For $\mathcal C\subset\mathcal N$, we define a corresponding default adjusted probability measure, denoted by  $\bar \P_\C$ and defined by:
\[
\frac{d\bar \P_\C}{d\P^{0}}\Big |_{\G_t^\mathcal N}= \prod_{k\in\mathcal C}\mathcal E_t(\nu(k)), \quad t\geq 0,
\]
 with $\nu(k)$ defined in (\ref{nuk}). 
The probability is well defined for all $t$ and all $\C$, as we have already assumed the processes $ \alpha$ and $\gamma$ to be bounded.
\end{defn}

We summarise the $(\GG^\mathcal C,\P^0)$ martingales that will play a role in the remaining:
\begin{align}\label{ni}
m_t(k)&:=\I_{\{\tau^A(k)\leq t\}}-\int_0^{t\wedge\tau(k)}\alpha_s(k)ds,\quad t\geq 0\\
n_t(k)&=\I_{\{T(k)\leq t\}}-\int_0^{t\wedge\tau(k)}\gamma_s(k)ds,\quad t\geq 0.
\end{align}

\subsection{Contagion via a change of the probability measure}


In order to introduce contagious impacts on the default intensities, we first define the following objects:
\begin{itemize}
\item [-] the direct  impact matrix $\bm \phi^A= (\phi^A_t(i,j))_{(i,j)\in \mathcal N^2}$ and 
\item[-] the indirect impact matrix $\bm \phi^B=(\phi^{B}_t(i,j))_{(i,j)\in \mathcal N^2}$,
\end{itemize}
with components being nonnegative and bounded processes that are  $\FF$-predictable. Here $\phi^A(i,j)$ (resp. $\phi^B(i,j)$) is the impact directly (resp. indirectly) induced by the default of the $j^{th}$ debtor on the $i^{th}$ debtor default intensity, whenever the last is not yet defaulted.

The following proposition is an application of the Girsanov's theorem.

\bigskip

\begin{prop}\label{PropChangeMeasure}
Let $\mathcal S$ be the set of contagious debtors, $\mathcal S\subset \mathcal N$. We introduce for all $i\in \mathcal S$ the predictable processes:
\begin{align}
A^\mathcal S_t(i):=\frac{1}{\alpha_t(i)}\sum_{j\in \mathcal S}\phi^A_{t}(i,j)\I_{\{\tau^A(j)< t\}},\quad t\geq 0\\
B^\mathcal S_t(i):=\frac{1}{\gamma_t(i)}\sum_{j\in \mathcal S}\phi^B_{t}(i,j)\I_{\{\tau^B(j)< t\}},\quad t\geq 0
\end{align}
whenever   $\alpha_t(i)>0$ resp. $\gamma_t (i)>0$;  and consider $A^\S_t(i)=0$ resp. $B_t^\S(i)=0$ otherwise. 

We define the family of probability measures $(\P^\mathcal S), \mathcal S\subset \mathcal N$:
\[
\frac{d\P^\mathcal S}{d\P^{0}}\Big |_{\G_t^\mathcal N}=D^\mathcal S_t:=\prod_{i\in \mathcal N} \mathcal E_t\left(\int_0^\cdot A_s^\mathcal S(i)dm_s(i)\right) \prod_{i\in \mathcal N} \mathcal E_t\left(\int_0^\cdot B^\mathcal S_s(i)dn_s(i)\right).
\]
Then, the default time $\tau(i)$, $i\in \mathcal N$ has the $(\GG^\mathcal N,\P^\mathcal S)$ intensity  given by:
\begin{align*}
\lambda^\mathcal S_t(i)&
=\lambda_t(i)+\left\{ \alpha_t(i)A^\mathcal S_t(i)+\beta_t(i)B^\mathcal S_t(i)\right\}.
\end{align*}
\end{prop}

\begin{rem} 
\begin{itemize}
\item[1.] We notice that the default intensities  under $\P^\S$ are of the form announced in (\ref{lambdaex}):
\begin{align*}
\lambda^\mathcal S_t(i)&=\lambda_t(i)+\sum_{j\in \mathcal C}\xi^{X(j)}_t(i,j)\I_{\{\tau(j)< t\}}\quad\text{for $i\in\N$,}
\end{align*} with $X(j)=A\I_{\{\tau(j)=\tau^A(j)\}}+B\I_{\{\tau(j)=\tau^B(j)\}}$, which is a $\G^\N_{\tau(j)}$ measurable random variable; and $\xi^A(i,j)=\phi^A(i,j)$ and $\xi^B(i,j)=g(i)\phi^B(i,j)$. 
\item[2.] Under $\P^\S$, some defaults may modify the evolution of the environment: the $(\GG^\N, \P^\S)$-intensity of a stopping time $T(i), i\in\N$ is
$\gamma(i)[1+B^\mathcal S_t(i)]$, i.e., has upward jumps at the default times $j\in\S$ that satisfy $\tau(j)=\tau^B(j)$. Or,  $T(i)_{i\in\N}$  are $\FF$-stopping times hence they are elements of the environment of the default system. 
\end{itemize}
\end{rem}

\section{Main result}\label{sec4}
We work under $(\Omega,\G,\GG^\N,\P^\N)$. We recall that under $\P^\N$ the class of contagious debtors is $\N$. This is without loss of generality: one can set the $k^{th}$ column of the two impact matrices $\bm \phi^A$ and $\bm \phi^B$ to be null and render the $k^{th}$ debtor non contagious. 

We want to characterise the time $t$ survival probabilities:
\[
\P^\N(\tau(k)>t, \forall k\in \C)\text{ for any }\C\in\N.
\]
We recall that under conditional independence,  the survival probabilities satisfy:
\begin{align*}
\P^0(\tau (k)>t,\; \forall k\in\C)& =\bar \E_\C \left[ \ell_t \right],
\end{align*}
where $\ell$ satisfies: $d\ell_t=-\ell_t( \sum_{k\in \C}\lambda_t(k))  dt$ (see the expression in (\ref{p0})). 
Our aim is to propose formulas under $\P^\N$ that have a similar form, that is:
\begin{align}\label{el}
\P^\N(\tau (k)>t,\; \forall k\in\C)& =\bar \E_\C \left[ \ell_t \right],
\end{align}
where  $\ell$ is an $\FF$ adapted process. But now, $\ell$ belongs to a larger family of processes that arises as solution of a system of linear stochastic differential equations that can be solved recursively. This is the object of Theorem \ref{MainThm} below, which is the main result of this paper. 

It would be tempting to denote the process $\ell$ appearing in (\ref{el}) $\ell^\C$, to reflect that it corresponds to the survival probabilities in the group $\C$. However, we refrain from doing so; instead our notation will be: $\ell=\ell^{\N-\C}$. We make the choice that subsets of $\N$ appearing as superscripts indicate the contagious entities. Indeed, we observe that:
\begin{align}\label{ell2}
\P^\N(\tau (k)>t,\; \forall k\in\C)& =\P^{\N-\C}(\tau (k)>t,\; \forall k\in\C),
\end{align}
i.e., we can consider that $\N-\C$ is in fact the set of contagious debtors when computing the above probability. This is because under $\P^\N$, the contagion produced by a particular debtor occurs only after its default and is inexistent before. In mathematical terms, the following Radon-Nikod\'ym density processes satisfy
\begin{equation}\label{DTk}
D^\N_t\I_{\{\tau (k)>t,\; \forall k\in\C\}}=D^{\N-\C}_t\I_{\{\tau (k)>t,\; \forall k\in\C\}}.
\end{equation}as resulting from the expressions in Proposition \ref{PropChangeMeasure}.
\bigskip

\begin{notation}
\begin{itemize}
\item[-]  Given a vector  $(V(i),i\in\N)$ and a matrix $M=(M(i,j),i,j\in\N)$ and with $\C,\D\subset \N$ we write
\[
 V(\C):=\sum_{i\in\C}V(i)\quad\text{and}\quad   M(\C,\D):=\sum_{i\in\C}\sum_{j\in\D} M(i,j).
\] For instance $\lambda_t(\C)=\sum_{i\in\C}\lambda_t(i)$ and $\phi^A_t(\C,\D)=\sum_{i\in\C}\sum_{j\in\D} \phi^A_t(i,j) $, etc.

\item[-] Whenever single elements $\{i\}$ of $\N$ appear as superscripts, we shall omit the brackets. That is: we write $\GG^i$ instead of $\GG^{\{i\}}$, $\GG^{\C\cup i}$ instead of $\GG^{\C\cup \{i\}}$ etc. 
\end{itemize}
\end{notation}
\bigskip

\begin{thm}\label{MainThm} Suppose that $\C,\D\in\N$  with $\C\cap\D=\emptyset$ and denote $\S:=\N-\C$. Then:
\begin{equation}\label{ellSD}
\P^\N(\tau(k)>t, \forall k\in \C\;;\; \tau^B(j)\leq t,\forall j\in\D)=\bar\E_\C\left[\ell^{\S|\D}_t\prod_{j\in\D}p_t(j)\I_{\{T(j)\leq t\}}\right],
\end{equation} where $\ell^{\S|\D}$ satisfies:
\begin{align}\label{L}
\nonumber
d\ell^{\S|\D}_t=
&\left\{-\ell^{\S|\D}_{t^-}\lambda_t(\C) - \sum_{j\in \S-\D}\left(\ell^{\S|\D}_{t^-}-\ell^{\S-j|\D}_{t^-}-\ell^{\S|\D\cup j}_{t^-} p_t(j)\I_{\{T(j)<t\}}\right)\psi_t^A(\C\cup\D,j)\right\} dt \\
&+\sum_{k\in \N} \left\{\sum_{j\in \S}\I_{\{T(j)<t\}}\left(\I_{\{j\in \D\}}  \ell^{\S|\D}_{t^-}+\I_{\{j\in\S-\D\}}  \ell^{\S|\D\cup j}_{t^-} p_t(j)\right)\frac{\phi^B_{t}(k,j)}{\gamma_t(k)}\right\}d n_t(k)\\\nonumber
 \ell^{\S|\D}_0=&1.
\end{align}
Above, we have denoted:
\[
\psi_t^A(k,j):=
\begin{cases} 
\phi_t^A(k,j) & k \in \C \\
\phi_t^A(k,j)\I_{\{T(k)>t\}} & k \in \D.
\end{cases}
\]
and
\[
p_t(k):=e^{-\int_0^{t\wedge T(k)}\alpha_s(k)ds}(1-e^{-\eta_{t\wedge T(k)}(k)}).
\]
In particular, denoting $ \ell^\S:= \ell^{\S|\emptyset}$, the survival probability in group $\C$ satisfies:
\begin{equation}\label{ellS}
\P^\N(\tau (k)>t,\; \forall k\in\C)=\bar \E_\C \left[  \ell ^\S_t \right],
\end{equation}with:
\begin{align}\label{LS}
\nonumber
d \ell^{\S}_t=
&\left\{- \ell^{\S}_{t^-}\lambda_t(\C) - \sum_{j\in \S}\left( \ell^{\S}_{t^-}- \ell^{\S-j}_{t^-} - \ell^{\S| j}_{t^-} p_t(j)\I_{\{T(j)<t\}}\right)\phi_t^A(\C,j)\right\} dt \\
&+\sum_{k\in \N} \left\{\sum_{j\in \S}\I_{\{T(j)<t\}} \ell^{\S| j}_{t^-} p_t(j)\frac{\phi^B_{t}(k,j)}{\gamma_t(k)}\right\}d n_t(k)\\\nonumber
 \ell^{\S}_0=&1.
\end{align}
\end{thm}

We postpone to Section \ref{sec5} the proof of this result. For now, we want to explore the SDEs above.

We begin by emphasising some particular cases:
\begin{enumerate}
\item \textbf{Conditional independence, no contagion.} If  $\bm\phi^A\equiv 0$ and $\bm \phi^B\equiv 0$ (i.e., there is no contagion), then $\P^\N=\P^0$ and:
\begin{align*}
d \ell^\S_t&=-  \ell^\S_{t}\lambda_t(\C)dt
\end{align*}which corresponds indeed to the expression in (\ref{p0}). In addition, by taking $\eta(i)\equiv 0$, we obtain that the default process $Y(i)$ of debtor $i$ is a Cox process. We refer to the case $\bm\phi^A= \bm \phi^B\equiv 0$ and $\eta(i)\equiv 0$ for all $i$ as the \textbf{Cox process setting}.

\bigskip 

\item \textbf{Conditionally Markovian setting with interacting intensities.}  If  $\eta(i)\equiv 0$ for all $i\in\N$, then also the following hold for all $i\in\N $:  (i) the hazard process $\Gamma(i)$ is continuous therefore  $\tau(i)$ avoids the $\FF$ stopping times, and (ii) $g(i)=0$ therefore the stopping time $\tau^B(i)=+\infty$, $\P^\N$- a.s. (its intensity is null). Consequently, In this case, there is no impact of the default system on its environment under the measure $\P^\N$. We recover in this way a Markovian framework similar to the one introduced in Section \ref{sec2}: the default indicator process $\Y$ is Markov, conditionally to $\F_\infty$, with transition rates at time $t$ from state $\mathbf x\in\{0,1\}^n$ to another state $\mathbf y\in\{0,1\}^n$ is:
\[
q_t( \mathbf x,\mathbf y)=
\begin{cases}
\lambda_t(k)+ \sum_{j\in\N} \phi^A_t( k, j) x(j) \text{ if }\exists \;  k\in\N:  \mathbf y=\mathbf x^k \text{ and }  x(k)=0 \\
0 \text{ else},
\end{cases}
\]where, as in the previous section,  $\mathbf x^k$ is obtained from $ \mathbf x=(x(1),\cdots, x(n))\in \{0,1\}^n$ by flipping the $k^{th}$ coordinate, $x(k)$.
 
We observe that $\bar \P_\C=\P^0$ and (\ref{LS}) becomes:
\begin{align}\label{LSMarkov}
d \ell^\S_t=&  - \ell^\S_{t}\left\{\lambda_t(\C)+\phi^A_{t}(\C,\S)\right\}dt+\sum_{j\in\S}    \ell^{\S-j}_{t}\phi^A_t(\C,j) dt.
\end{align}
  The formula (\ref{LSMarkov}) can also be obtained directly from the Kolmogorov forward equations associated with the default process $Y$. The interested reader can find the details in Appendix \ref{AppendixC}. 
  
\bigskip 
  
\item \textbf{Non-Markovian setting.} This case is obtained whenever indirect contagious entities exist, that impact the environment by their default. Any entity $j$ satisfying  $\eta(j) \neq 0$, $i \in \N$, is able to impact the environment, provided the impact matrices $\bm \phi^B$  has non zero elements in column $j$. We distinguish two cases:

\begin{itemize}
    \item \textit{Indirect contagion only.} If  $\bm \phi^A\equiv 0$ and $\bm \phi^B\neq 0$ (i.e., there is only indirect contagion), then:
\begin{align*}
d \ell^\S_t=&  - \ell^\S_{t^-}\lambda_t(\C)dt+\sum_{j\in \S} \I_{\{T(j)<t\}}  \ell^{S|j}_{t^-} p_t(j)  \sum_{k\in \N}  \left(   \frac{\phi^B_{t}(k,j)}{\gamma_t(k)}\right) dn_t(k).
\end{align*}
\item \textit{Indirect and direct contagion.} This is the general case, where $\bm \phi^A, \bm \phi^B\neq 0$, then there is both direct and indirect contagion and the evolution of $\ell^S$ is given by \eqref{LS}.
\end{itemize}
\end{enumerate}

We now indicate how one can  concretely obtain the survival probabilities from the SDEs in Theorem \ref{MainThm}. A target set $\C^*\subset \N$ is fixed and let $\S^*=\N-\C^*$. We want to obtain the process $ \ell^{\S^*}$. We proceed by iteration,  starting with $\S=\emptyset$ we recursively add elements so to create all possible subsets of $\S^*$. The set $\S^*=\N-\C^*$ is obtained at the last iteration.  More precisely, this works as follows:
\begin{enumerate}
\item[0.] $\S=\emptyset$. We compute $ \ell^\emptyset_{t}$.
\item[1.] For all $j\in \S^*$, we take $\S=\{j\}$ and obtain the quantities $ \ell^{j|j}$ and $ \ell^{j}$.
\item[2.]  For all $\{j_1,j_2\}\subset \S^*$, we take $\S=\{j_1,j_2\}$ and obtain the quantities $ \ell^{S|S}, \ell^{S|j_1}, \ell^{S|j_2}, \ell^{\S}$ (in that order).

\item[...]

\item[]In general, at the $k^{th}$ iteration:
\item[k.] For any $\S\subset \S^*$ with $\mathbf{card}(\S)=k$ and for any $\D\subset \S$, we compute $ \ell^{\S|\D}$, in the decreasing order of the cardinality of $\D$.  There are $\binom{n}{k}$ subsets of $\S^*$ that contain $k$ elements, each of them having $2^k$ different subsets. Hence, at the $k^{th}$ iteration, we have to solve $\binom{n}{k} 2^k$ equations of the type  (\ref{L}). For solving these equations, the quantities obtained at step $k-1$ are needed.
\end{enumerate}

For instance, if $\mathbf{card}(\S^*)=s$, the procedure necessitates iterations $0,1,\cdots, s$ of the form described above, that is, we need to solve for:
\[
\sum_{k=0}^s{\binom{s}{k}} 2^k=3^s
\]equations of the type (\ref{L}). We see that the complexity of the procedure is high, when applied to default systems of big size, which is also a typical feature of the Markovian framework, where one would need $2^s$ iterations. 

In practical applications however, we advocate that the complexity can be reduced as follows. In most financial systems, even though there are a multitude of debtors, the number of those defaults that are expected to have a notable impact outside the system itself is presumably limited to a few entities (the systemic firms). The other firms can be considered as non systemic: we can assume that $\eta(k)=0$, that is,  $\tau^B(k)=\infty$ $ a.s.$ when debtor $k$ is non systemic. The interpretation is that if debtor $k$ is not systemic, its default has at most a direct contagious impact on its counterparties (i.e., the other debtors in the default system), but not a larger economic impact (i.e. on the environment of the default system).

For example, suppose that $\S^*=\S^*_A\cup\S^*_B\subset \N$ and for all $k\in\S^*_A$ we have $\tau^B(k)=\infty$ $ a.s.$, that is $\tau(k)=\tau^A(k)$ $a.s.$. In other words, $\S^*_A$ is a group of non systemic debtors and $\S^*_B$ contains possibly systemic debtors. We consider $\S^*_A\cap \S^*_B=\emptyset$ and $\mathbf{card}(\S^*_B)=b$, so that $\mathbf{card}(\S^*_A)=s-b$. 
In order to obtain the process  $ \ell^{\S^*}$, we need this time to solve for:
\[
\sum_{k=0}^{s-b}{\binom{s-b}{k}}\times\left(\sum_{k=0}^b{\binom{b}{k}} 2^k\right)=2^{s-b}3^b
\]equations of the type (\ref{L}), hence a reduced complexity. 
 

\section{Numerical implementation}
\label{sec:numerical_implementation}

For the sake of numerical implementations, we specialize our model to the following setup.

We assume that the filtration $\FF$ is generated by a one-dimensional factor process $\Psi$, which we take to be a basic affine jump diffusion (as introduced in \cite{DufGar01}) under $\P^0$, with dynamics given by:

\begin{equation}
\label{eq:psi}
d \Psi_t = \alpha (b - \Psi_t) dt + \sigma \sqrt{\Psi_t} dW_t + dJ_t,
\end{equation}
where $\alpha, b, \sigma > 0$ and are such that $\alpha b \ge \frac{1}{2}\sigma^2$, $W$ is a standard Brownian motion and $J$ is a pure jump process with constant jump intensity $\tilde{\gamma} > 0$ and exponentially distributed jump sizes with expected value $\mu > 0$.

The random times $(T(k))_{k \in \N}$ are a subset of the jump times of the factor process $\Psi$, obtained via the recursive thinning procedure in Algorithm \ref{algo:thinning}, which guarantees that all $T(k)$ have identical and constant pre-jump intensities given by:

\begin{equation}
\label{eq:gamma}
\gamma_t(k) = \tilde{\gamma} \tilde{\pi} \I_{\{T(k) > t\}},    
\end{equation}
for a given thinning parameter $\tilde{\pi} \in (0, 1/n]$.

\bigskip
\begin{algorithm}[H]
\RestyleAlgo{ruled}
\KwIn{Ordered jump times $(\sigma(j), j \in \NN)$ of factor $\Psi$; thinning parameter $\tilde{\pi}$.}
\KwOut{Jump times $(T(k), k \in \C)$.}
$T(k) := \infty, \forall k \in \N$\;
$\sigma(0) := 0$\;
\For{$j = 1$ \KwTo $\infty$}{
    $U := \{k \in \N \text{ such that } T(k) > \sigma(j-1)\}$\;
    \eIf{$U \neq \emptyset$}{
        with probability $\tilde{\pi} \cdot |U|$ sample $k^*$ uniformly at random from $U$\;
        $T(k^*) := \sigma(j)$\;
    }{
        break\;
    }
}
\caption{Thinning procedure}
\label{algo:thinning}
\end{algorithm}
\bigskip

Finally, we choose constant impact matrices: 
\begin{equation*}
\Phi^A_t(i,j) = \phi^A > 0, \quad \Phi^B_t(i,j) = \phi^B > 0,
\end{equation*}
a linear factor dependence for the specific default intensity 
\begin{equation}
\label{eq:alpha}
\alpha_t(k) = \lambda_1 \cdot \Psi_t + \lambda_0, \forall k \in \N, \text{ with } \lambda_0, \lambda_1 > 0,
\end{equation}
and a homogeneous and deterministic jump for all hazard processes
\begin{equation*}
\eta_t(k) := \eta > 0.
\end{equation*}

\begin{rem}
We point out that our model can accommodate much more complex setups than the one just presented. For instance, the times $(T(k), k \in \N)$ don't necessarily have to correspond to jump times of the factor $\Psi$. More specifically, if we take the filtration $\FF$ to be $\F_t = \sigma(\Psi_s, (\I_{\{T(k) < s\}}, k \in \N), s \le t)$, then the random times $T(k)$ can model arbitrary systemic events (not just factor jumps) which can affect some (or all) debtors via their corresponding $\FF$-predictable processes (e.g. $\alpha$ and $\eta$, or the impact matrices $\bm \phi^A$ and $\bm \phi^B$).
\end{rem}

\subsection{Term structure of default probabilities}
\label{subsec:PDs}

Given a set of debtors $\N$ and a fixed time horizon $T$, we can use Theorem \ref{MainThm} to compute the term structure of the joint survival probability $\{\P(\tau(k) > t, \forall k \in \C^*), t \in [0, T]\}$ of a target set $\C^* \subseteq \N$.

The expectation of $\ell^\S$ under the probability measure $\bar{\P}$ is obtained via Monte Carlo estimation by sampling $\ell^\S$ first under $\P^0$ and then multiplying it path-wise by the change of measure in Definition \ref{PC}. Samples of $\ell^\S$ under $\P^0$ can be obtained from \eqref{LS} by solving recursively the SDEs in \eqref{L}, as explained in the remarks following Theorem \ref{MainThm}. This sampling procedure is explained in more detail for our setup in Algorithm \ref{algo:LS}.

\bigskip
\begin{algorithm}[H]
\RestyleAlgo{ruled}
\KwIn{target set $\C^* \subseteq \N$; parameters of $\Psi$ ($\alpha, b, \sigma, \tilde{\gamma}, \mu$); thinning parameter $\tilde{\pi}$; contagion parameters ($\phi^A, \phi^B, \lambda_0, \lambda_1$, $\eta$); time grid $(t_i)_{i=0:M}$.}
\KwOut{sample path $(\ell^\S_{t_i})_{i=0:M}$ under $\bar{\P}$}

    sample $(\Psi_{t_i})_{i=0:M}$ from \eqref{eq:psi} under $\P^0$ using the Euler-Maruyana scheme\;
    extract $(T(k), k \in \N)$ from the factor process jumps using Algorithm \ref{algo:thinning}\;
    \For{$k \in \N$}{
        compute $\alpha(k)$ from \eqref{eq:alpha}, $\gamma(k)$ from \eqref{eq:gamma} and $g(k)$ from \eqref{eq:g}\;
        compute $p(k)$ from \eqref{eq:p}, $\lambda(k)$ from \eqref{eq:lambda} and $n(k)$ from \eqref{eq:n}\;
    }
    set $S^* := \N - \C^*$\;
    \For{$S \textbf{ in } {\sc Subsets}(S^*)$ in increasing order of cardinality}{
        \For{$D \textbf{ in } {\sc Subsets}(S)$ in decreasing order of cardinality}{

            Compute $\ell^{\S|\D}$ as in \eqref{L}.
            
        }
    }
    \Return{$(\ell^{\S|\emptyset}_{t_i} \cdot \prod_{k \in \C^*} \mathcal{E}_{t_i}(\nu(k))_{i=0:M}$}
\caption{Sampling $\ell^\S$ under $\bar{\P}$}
\label{algo:LS}
\end{algorithm}
\bigskip





We investigate numerically the behavior of the model for a homogeneous group of $n=5$ debtors as the contagion parameters $\phi^A$, $\phi^B$, $\eta$ are varied. In particular we are interested in comparing the survival probability of the first debtor, $\P(\tau(1) > t)$, over a time horizon of $T = 10$ years under three contagion settings: the Cox process setting ($\eta=0$, $\phi^A=0$, $\phi^B=0$), the interacting intensities setting ($\eta=0$, $\phi^A>0$, $\phi^B=0$) and the non-Markovian setting ($\eta>0$, $\phi^A>0$, $\phi^B>0$). 

\begin{figure}[H]
    \includegraphics[width=0.8\textwidth]{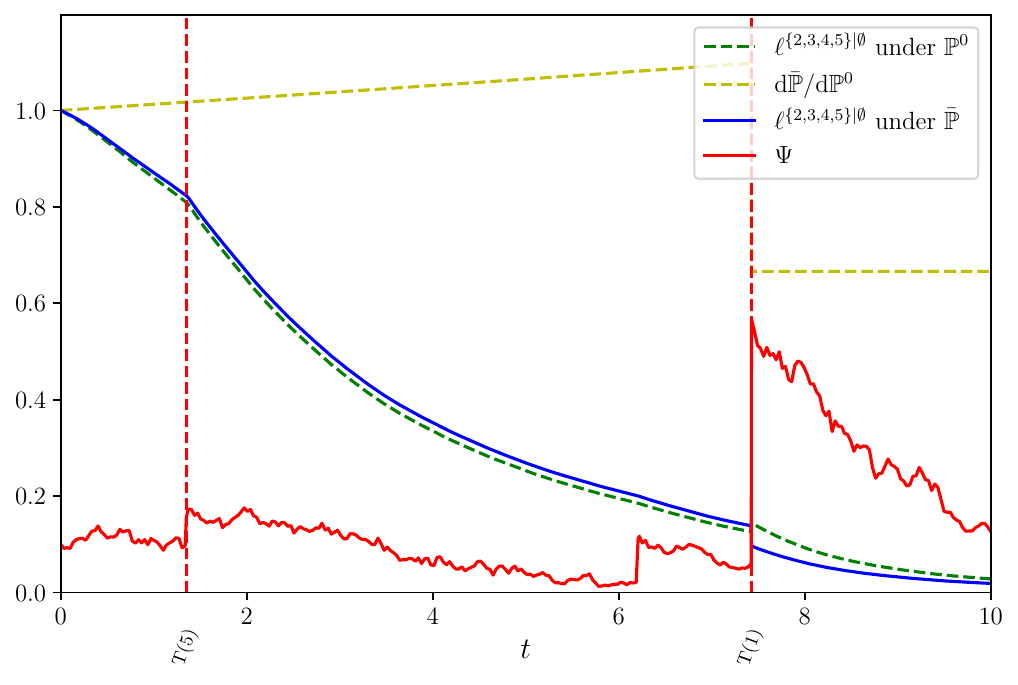}
    \caption{Simulated sample paths of factor process $\Psi$, process $\ell^{S|\emptyset}$ and the change of measure density process from a single run of Algorithm \ref{algo:LS}.}
    \label{fig:sample_path}    
\end{figure}

The factor process follows the dynamics in \eqref{eq:psi} with parameters $\alpha = 0.6, \mu = 0.1, b = 0.02, \sigma = 0.14, \tilde{\gamma} = 0.2$ and initial value $\Psi_0 = 0.1$. We set $\tilde{\pi} = 0.8 \cdot n$ and $\lambda_1 = 1, \lambda_0 = 0$ for simplicity. The process is simulated using an Euler-Maruyama scheme with $300$ steps. The survival probability is estimated using Monte Carlo estimation on $10'000$ samples obtained via Algorithm \ref{algo:LS}. A sample path is shown in Figure \ref{fig:sample_path}, together with the realization of the factor process and the change of measure. The model has been implemented in Python and is available online at \url{https://github.com/gvisen/overspilling-contagion}.

\begin{figure}[H]
    \includegraphics[width=0.8\textwidth]{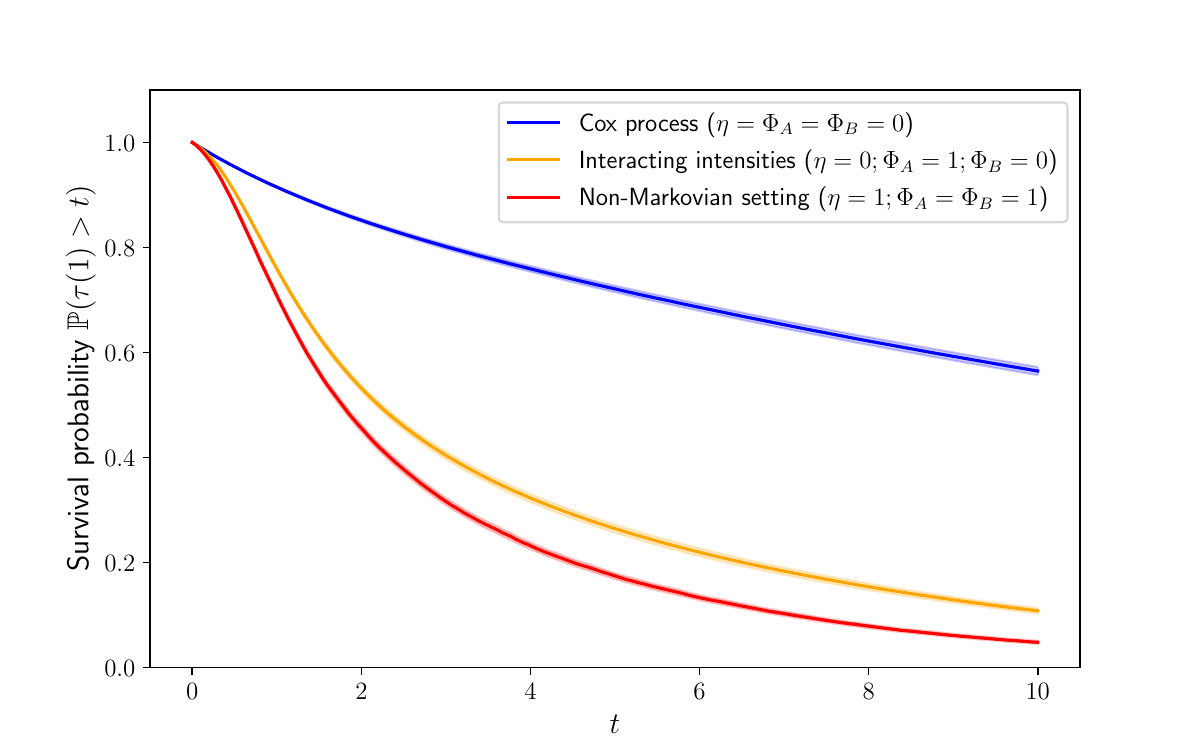}
    \caption{\footnotesize Monte Carlo estimate of survival probability $\P(\tau(1) > t)$ under three representative contagion settings ($10'000$ MC samples, asymptotic $99.99\%$ confidence band).}
    \label{fig:comparison_of_cases}    
\end{figure}

Figure \ref{fig:comparison_of_cases} shows the survival probability of a single debtor in the group under different contagion settings. In particular, we notice that in the non-Markovian setting the addition of an indirect contagion mechanism on top of the direct one determines a downward shift in the default curve in the medium and long terms.

\begin{figure}[H]
    \includegraphics[width=0.8\textwidth]{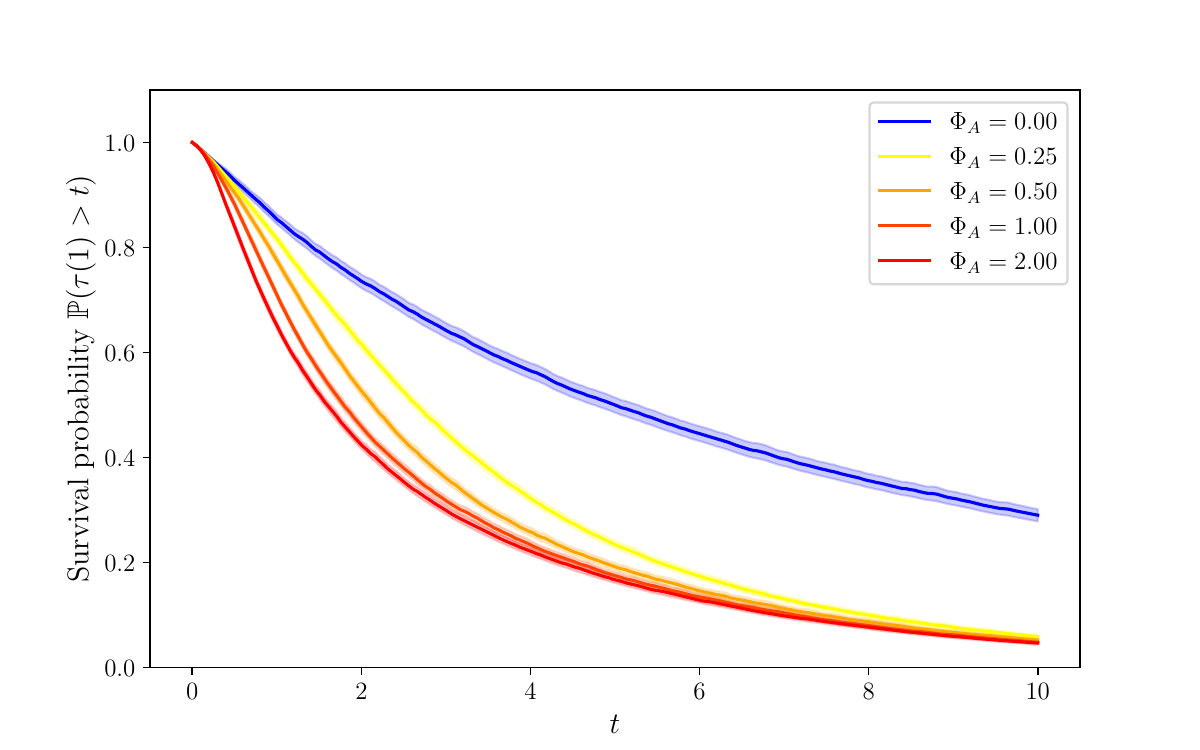}
    \caption{\footnotesize Monte Carlo estimate of survival probability $\P(\tau(1) > t)$ in the non-Markovian setting ($\eta=1, \phi^B=1$) for several values of $\phi^A$ ($10'000$ MC samples, asymptotic $99.99\%$ confidence band).}
    \label{fig:varying_phiA}
\end{figure}

Figure \ref{fig:varying_phiA} shows the impact of the direct contagion parameter $\phi^A$, whose effect approaches a long-term saturation point for values close to $2$. A similar effect is noticeable for the parameter $\phi^B$, which reaches saturation for values close to 2, as shown in Figure \ref{fig:varying_phiB} (for two different choices of $\phi^A$).

\begin{figure}[H]
\centering
\begin{subfigure}{0.8\textwidth}
  \centering
  \includegraphics[width=\textwidth]{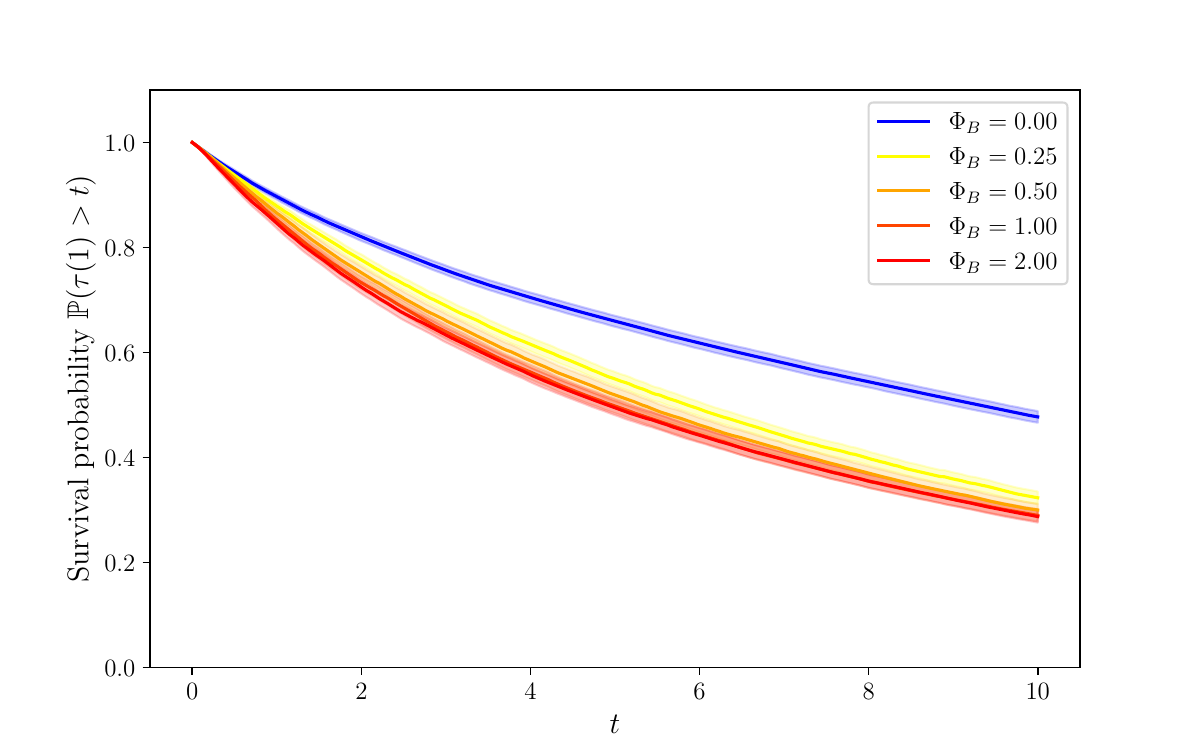}
  \caption{}
  \label{fig:sub1}
\end{subfigure}
\\ 
\begin{subfigure}{0.8\textwidth}
  \centering
  \includegraphics[width=\textwidth]{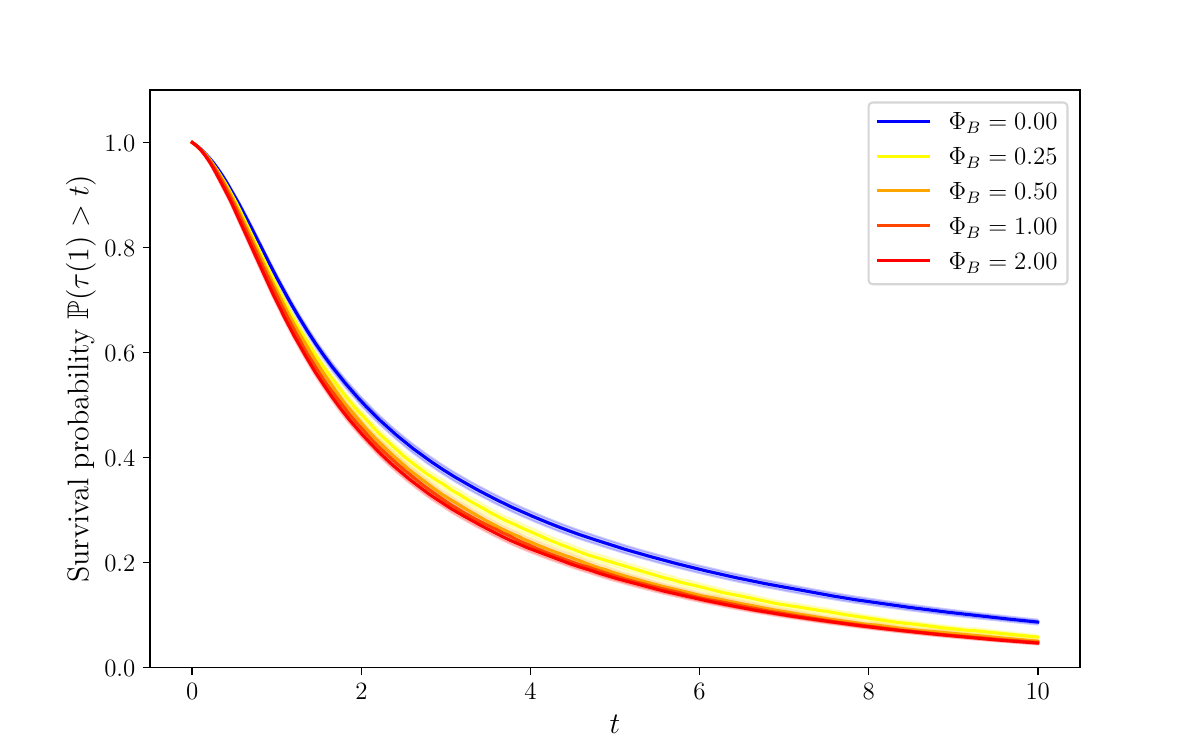}
  \caption{}
  \label{fig:sub2}
\end{subfigure}
\caption{\footnotesize Monte Carlo estimate of survival probability $\P(\tau(1) > t)$ in the non-Markovian setting (panel (A): $\eta=1, \phi^A=0$; panel (B): $\eta=1, \phi^A=1$) for several values of $\phi^B$ ($10'000$ MC samples, asymptotic $99.99\%$ confidence band).}
\label{fig:varying_phiB}
\end{figure}

\subsection{Pricing of single-name CDS}
\label{subsec:CDS}

In a single-name credit default swap (CDS) the buyer makes regular premium payments to the seller in exchange for its commitment to cover losses caused by the default of a reference entity, whenever they materialize, before the end of the contract. The pricing of a single-name CDS therefore involves the evaluation of the expected values of the premium payments leg and the default payment leg.

The premium payments are made at regular times, say $0 = t_0 < t_1 < \ldots < t_N = T$, and are expressed in terms of an annualized spread, here denoted by $x$. The premium payment at time $t_n$ is then equal to $(t_n - t_{n-1}) x$. 

If the reference entity defaults at a random time $\tau \in (t_{n-1}, t_n]$, then the buyer is also required to pay the premium accrued since the last premium payment, i.e. $x(\tau - t_{n-1})$. The expected discounted value at time $t=0$ of the premium payments leg under an equivalent martingale measure is given by:

\begin{align}
\label{eq:premium_leg}
L_{\text{prem}}(s) & = \mathbb{E} \left[ \sum_{k = 1}^N \left( e^{- \int_0^{t_k} r(u) du} s (t_k - t_{k-1}) \I_{\{\tau > t\}} + e^{- \int_0^{\tau} r(u) du} s (\tau - t_k) \I_{\{t_{k-1} < \tau \le t_k\}} \right) \right] \nonumber \\
& = \sum_{k = 1}^N \left( e^{- \int_0^{t_k} r(u) du} s (t_k - t_{k-1}) \mathbb{P}(\tau > t) + \int_{t_{k-1}}^{t_k} e^{- \int_0^{t} r(u) du} s (t - t_k) f_\tau(t) dt \right) \nonumber \\
\end{align}

where $f_\tau$ is the density function of $\tau$ and $r$ is the risk-free interest rate, here assumed deterministic\footnote{The assumption of a deterministic risk-free interest rate is commonly made in credit risk, since incorporating interest rate risk in practice leads to negligible contributions compared to the high uncertainty due to default risk.}.

The CDS seller makes at the random time $\tau$. The expected discounted value at time $t=0$ of the default cashflow is therefore:

\begin{align}
\label{eq:default_leg}
L_{\text{def}} & = \mathbb{E} \left[ e^{- \int_0^\tau r(u) du} \delta \I_{\{t < \tau \le T\}} \right] \nonumber \\
& = \delta \int_0^T e^{- \int_0^{t} r(u) du} f_\tau(t) dt \nonumber \\
\end{align}

where $\delta$ is the actual default payment, here assumed to be constant.

The fair spread of a CDS contract can then be computed by solving for $s$ in the equation $L_{\text{prem}}(s) = L_{\text{def}}$.

Table \ref{table:CDS} shows the estimated fair spreads for a single-name CDS contract evaluated using our model under risk-neutral parameters for different contagion settings. 

In order to allow comparison with the results shown so far, the reference entity is chosen to be one of the debtors in the group of debtors presented in Section \ref{subsec:PDs}. All model implementation parameters have been kept identical. The premium and default legs are computed from the distribution of the default time $\tau(1)$. We further assumed $r=2\%$ and $\delta = 60\%$. 

As Table \ref{table:CDS} shows, the introduction of indirect contagion in the non-Markovian setting leads to higher spreads and can be used to capture the impact of the default of systemically important institutions on the credit spread of a debtor.

\begin{table}[H]
\centering
\begin{tabular}{l|p{0.6cm}|p{0.6cm}|c|}
    \hline
     & $\phi^A$ & $\phi^B$ & Spread (bp) \\
     \hline \hline
     Cox process & & & \\ \hline
      & - & - & 18.19 $\pm$ 0.20 \\
     \hline \hline
     Interacting intensities & & & \\ \hline
      & 1.0 & - & 80.09 $\pm$ 0.82 \\
      & 2.0 & - & 85.77 $\pm$ 0.92 \\
     \hline \hline
     Non-Markovian setting & & & \\ \hline
     \hfill $\eta = 1.0$ & 1.0 & 1.0 & 109.13 $\pm$ 0.93 \\
      & 1.0 & 2.0 & 111.55 $\pm$ 1.09 \\
      & 2.0 & 1.0 & 116.47 $\pm$ 0.98 \\
      & 2.0 & 2.0 & 118.76 $\pm$ 1.16 \\ \hline
     \hfill $\eta = 2.0$ & 1.0 & 1.0 & 121.21 $\pm$ 1.01 \\
      & 1.0 & 2.0 & 124.04 $\pm$ 1.24 \\
      & 2.0 & 1.0 & 129.74 $\pm$ 1.34 \\
      & 2.0 & 2.0 & 131.73 $\pm$ 1.23 \\
    \hline
\end{tabular}
\caption{\footnotesize Monte Carlo estimates of fair spreads (in basis points, bp) for a single-name CDS contract under different contagion settings. Estimates obtained from $10'000$ samples of the survival probability from Algorithm \ref{algo:LS} and reported with their $99\%$ asymptotic confidence interval.}
\label{table:CDS}
\end{table}

\subsection{Pricing of \texorpdfstring{$k$th-to-default swaps}{kth-to-default swaps}}
\label{subsec:kth-to-default}

A $k$th-to-default swap is a basket credit derivative in which the buyer is entitled to receive a default payment at the time of the $k$-th default in a reference portfolio. In exchange for this payment the buyer makes regular premium payments to the seller. Assuming a constant and identical recovery rate, the pricing of a $k$th-to-default swap is identical to the pricing of a single-name CDS, as done in Section \ref{subsec:CDS}, provided the default time $\tau$ is substituted with the $k$th-to-default time.

The distribution of the $k$th-to-default time can be obtained from the joint survival probabilities of all subsets of debtors in the reference portfolio using a very simple recursive scheme. If we denote the $k$th-to-default time by $\tau^k$ (for $k \ge 1$), then we have that $\tau^k > t$ if and only if at most $k-1$ entities have defaulted by time $t$, that is

\begin{equation}
\label{eq:2}
    \P(\tau^k > t) = \sum_{J \subseteq \N, |J| < k} \P(\tau(k) > t, k \in \N-J; \tau(j) \le t, j \in J).
\end{equation}

The terms in the summation can be computed starting from the following trivial equation:

\begin{align*}
    \P(\tau(k) > t, \forall k \in N - J) & = \sum_{I \subseteq J} \P(\tau(k) > t, k \in N-I; \tau(i) \le t, \forall i \in I),\\
\end{align*}

which can be expressed more compactly as follows:

\begin{equation}
\label{eq:2bis}
p_{N-J, \emptyset}(t) = \sum_{I \subseteq J} p_{N-I, I}(t),    
\end{equation}

by introducing the convenient notation $p_{A,B}(t) := \P(\tau(k) > t, k \in A, \tau(j) \le t, j \in B)$, for $A \cap B = \emptyset$.

Solving for the term corresponding to $I = J$ in the right-hand side summation, one obtains

\begin{equation}
p_{N-J, J}(t) = p_{N-J, \emptyset}(t) - \sum_{I \subset J, I \neq J} p_{N-I, I}(t),
\end{equation}

which can be solved recursively on $J$, provided the terms $p_{N-J, \emptyset}(t)$ (i.e. the joint survival probabilities of all subsets of debtors) are known.

In the case of our model these terms can be computed via Monte Carlo estimation using Algorithm \ref{algo:LS} for $\ell^{J|\emptyset}$.

\begin{rem}
As clear from Algorithm \ref{algo:LS}, each sample path of $\ell^{A|\emptyset}$ requires computing $\ell^{B|\emptyset}$ for all $B \subseteq A$, therefore the Monte Carlo samples used to estimate the $k^*$th-to-default time for a given $k^*$ can be used to produce an estimate (albeit not an independent one) of the $k$th-to-default times, for all $k \le k^*$. 
\end{rem}

\begin{rem}
If entities in the reference portfolio are assumed to be exchangeable\footnote{This is a reasonable assumption for all homogeneous portfolios and in particular for all major credit indices, such as indices of the CDX.NA.IG and iTraxx Europe families.}, then the terms $p_{A, B}(t)$ depend only on the cardinality of the sets $A$ and $B$ and the recursion requires only $k$ steps. It then follows that the pricing of a $k$th-to-defeault swap for a portfolio of $n$ debtors requires solving $O(3^k)$ SDEs, independently of the number of debtors in the reference portfolio.
\end{rem}

Table \ref{table:kth-to-default} shows the fair spreads for various values of $k$. Also in this case, we notice that the addition of indirect contagion in the non-Markovian setting leads to higher spreads.

\begin{table}
\centering
\begin{tabular}{l|p{0.6cm}|p{0.6cm}|c|c|c|}
    \hline
     & $\phi^A$ & $\phi^B$ & $k=1$ & $k=2$ & $k=3$ \\
     \hline \hline
     Cox process & & & & & \\ \hline
      & - & - & 92.09 $\pm$ 1.14 & 33.13 $\pm$ 0.48 & 14.27 $\pm$ 0.27 \\
     \hline \hline
     Interacting intensities & & & & & \\ \hline
      & 1.0 & - & 90.52 $\pm$ 1.02 & 83.50 $\pm$ 0.89 & 79.35 $\pm$ 0.81 \\
      & 2.0 & - & 91.20 $\pm$ 1.05 & 87.23 $\pm$ 0.97 & 84.78 $\pm$ 0.93 \\
     \hline \hline
     Non-Markovian setting & & & & & \\ \hline
     \hfill $\eta = 1.0$ & 1.0 & 1.0 & 120.27 $\pm$ 1.66 & 108.06 $\pm$ 1.45 & 101.67 $\pm$ 1.36 \\
      & 1.0 & 2.0 & 120.18 $\pm$ 1.86 & 108.80 $\pm$ 1.64 & 102.59 $\pm$ 1.57 \\
      & 2.0 & 1.0 & 121.13 $\pm$ 2.03 & 113.47 $\pm$ 1.86 & 109.62 $\pm$ 1.78 \\
      & 2.0 & 2.0 & 119.84 $\pm$ 1.78 & 112.98 $\pm$ 1.64 & 109.29 $\pm$ 1.58 \\ \hline
     \hfill $\eta = 2.0$ & 1.0 & 1.0 & 131.78 $\pm$ 2.25 & 117.82 $\pm$ 1.95 & 110.37 $\pm$ 1.80 \\
      & 1.0 & 2.0 & 131.66 $\pm$ 2.45 & 118.02 $\pm$ 2.13 & 110.64 $\pm$ 2.01 \\
      & 2.0 & 1.0 & 132.15 $\pm$ 2.21 & 123.83 $\pm$ 2.01 & 119.30 $\pm$ 1.93 \\
      & 2.0 & 2.0 & 131.98 $\pm$ 2.11 & 123.98 $\pm$ 1.95 & 119.50 $\pm$ 1.87 \\
    \hline
\end{tabular}
\caption{\footnotesize Monte Carlo estimates of fair spreads (in basis points, bp) for $k$th-to-default swap contracts under different contagion settings. Estimates obtained from $10'000$ samples of the probabilities in Equation \eqref{eq:2} and reported with their $99\%$ asymptotic confidence interval.}
\label{table:kth-to-default}
\end{table}


\section{Proof of the main result}\label{sec5}
This section is dedicated to the proof of the Theorem \ref{MainThm}. For the convenience of the reader, we gather separately, in Appendix \ref{Appendix} the basic results from the theory of the enlargement of filtrations that were useful for our proofs. Also for the sake of clarity, we establish some intermediary results in the first two subsections. 

The proof rely on projections in some subfiltrations of $\GG^\N$. For any set $\C\subset \N$, we introduce the filtration $\GG^{\mathcal C}$ as
\[
\G^\C_t:= \mathcal H^\C_{t^+} \text{ with } \mathcal H^{\C}_t:=\F_t\bigvee_{k\in \mathcal C} \sigma(t\wedge\tau(k)),
\]i.e., the progressively enlarged filtration that satisfies the usual conditions and makes any $\tau(k)$ with $k\in \mathcal C$ a stopping time. We have $\GG^\emptyset =\FF$ and $\GG^\N$ is as in (\ref{int}).

\subsection{Preparatory results (I)} Because we are dealing with several filtrations and probabilities, we clarify here what a martingale becomes when we change the filtration and/or probability. Only the relevant changes of filtration and probability are emphasised.

\textbf{Notation.} Given two filtrations $\FF\subset \GG$ and a probability measure $\P$, we write $\FF \overset{\P}{\hookrightarrow} \GG$ when all  $\FF$ martingales remain $\GG$ martingales under the probability measure $\P$. This property is usually called immersion property (i.e., we say that $\FF$ is immersed in $\GG$) or (H) hypothesis.

\begin{lem}\label{filtrations} Let $\S$ be a subset of $\N$. The following hold:
\begin{itemize}
\item[(a)]  
\[
\FF  \overset{\P^0}{\hookrightarrow}\GG^\S \overset{\P^0}{\hookrightarrow} \GG^{\N},
\]

\item[(b)] 
\[
\FF  \overset{\P^{\N}}{\not \hookrightarrow}\GG^{\S} \overset{\P^{\N}}{\not\hookrightarrow} \GG^{\N},
\]

\item[(c)] Under $\bar\P_\C$, where $\C\subset \N$, we have:
\[
\FF  \overset{\bar\P_\C}{\hookrightarrow}\GG^\S \overset{\bar \P_\C}{\hookrightarrow} \GG^{\N}.
\]
\end{itemize}
\end{lem}
\proof
\begin{itemize}
\item[(a)]   Let us consider $X\in\G^\S_\infty$. We denote $\mathcal H_t:=\G^\S_t\vee_{k\in\N-\S}\sigma(e(k))$. We have that $\G^\N_t\subset \mathcal H_t$ and because any $e(k), k\in\N-\S$ is independent from $\G^\S_\infty$, we obtain:
\[
\E^{0}[X|\G^\N_t]=\E^{0}[\E^0[X|\mathcal H_t] \G^\N_t]= \E^{0}[\E^0[X|\G^\S_t] \G^\N_t]=\E^0[X|\G^\S_t].
\] 
To conclude, we apply Theorem \ref{lecasdregeneral} (3).
\item[(b)] Under $\P^{\N}$  the $\GG^\N$-compensators of the $\FF$-stopping times $T(k),k\in\N$ are not adapted to any sub-filtrations of  $\GG^\N$, which proves the claim. 
 
  \item[(c)] The Radon-Nikod\'ym density process $\frac{d\bar \P_\C}{d\P^{0}} |_{\G_t^\mathcal N}$ is $\FF$ adapted. Then, by Proposition \ref{proplL}, the immersion property holds in the given superfiltrations of  $\GG^{\S}$, as it was holding under $\P^0$. 
\end{itemize}
\finproof

The results in Theorem \ref{MainThm} make appear expectations under $\bar\P_\C$.  Let us fix a set $\C\subset\N$. Under $\bar \P_\C$ we have that:
 \begin{itemize}
  \item[-] For $k\in \mathcal C$, the stopping time $T(k)$ has an $\FF$-intensity $\gamma (k)(1-g(k))$. We define the following $\bar\P_\C$-martingales (relative to the filtrations $\FF$ and $\GG^\N$):
  \begin{equation}\label{nc}
  \bar n^\mathcal C_t(k):=\I_{\{T(k)\leq t\}}-\int_0^{t\wedge T(k)}\gamma_s (k)(1-g_s(k))ds
  \end{equation}
   \item[-] For $k\in \mathcal N-\mathcal C$, the stopping time $T(k)$ has unchanged $\FF$-intensity $\gamma(k)$, i.e., the same as under $\P^0$. 
  \item[-] More generally, all the $\P^0$-martingales orthogonal to $n(k), k\in \C$ are also $\bar\P_\C$ martingales.
  \item[-] As seen in Lemma \ref{filtrations}, the  $(\FF,\bar \P_\C)$ martingales, in particular  $\bar n^\mathcal C(k)$, remain martingales in larger filtrations $\GG^\S$, $\S\subset\N$.
  \end{itemize}


\subsection{Preparatory results (II)} In this section, a set $\C$ is fixed, $\C\subset\N$,  and we consider two additional sets: $$\S:=\N-\C\text{ and }\D\subset\S.$$ 

In Theorem \ref{MainThm}, the SDE (\ref{L}) for $ \ell^{\S|\D}$  is obtained after projecting on the filtration $\FF$ a $\GG^{\S-\D}$ adapted process, that we shall denote $L^{\S|\D}$. In this section, we identify the process $L^{\S|\D}$ (in Proposition \ref{MainProp}) and prepare the building blocks for obtaining its $(\FF,\bar\P_\C)$ projection (Proposition \ref{BB}). 

\begin{prop}\label{MainProp} The following hold:
\begin{itemize}
\item[(a)]
\begin{align}\label{pc1}
\P^\N(\tau (k)>t,\; \forall k\in\C)& = \bar \E_\C[L_t^\S],
\end{align}
where 
\begin{align*}
L_t^{\mathcal S}&:=\exp \left (-\Lambda^{\mathcal S}_t( \mathcal C) \right)\underset{i\in \mathcal S}{\prod} \mathcal E_t\left(\int_0^tA_s^{\mathcal S}(i)dm_s(i)\right) \underset{i\in \mathcal N}{\prod}  \mathcal E_t\left(\int_0^tB^{\mathcal S}_s(i)d\bar n^\mathcal C_s(i)\right).
\end{align*}
\item[(b)]
\begin{align}\label{pc2}
\P^\N(\tau (k)>t,\forall k\in\C\;;\;\tau^B(j)\leq t,\forall j\in\D)& =  \bar \E_\C\left [L_t^{\S|\D}\prod_{j\in\D}p_t(j)\I_{\{T(j)\leq t\}}\right],
\end{align}
where 
\begin{align}\label{LSDformula}
L_t^{\S|\D}:=&\underset{i\in \S-\D}{\prod} \mathcal E_t\left(\int_0^\cdot A_s^{\S-\D}(i)dm_s(i)\right)\underset{i\in \N}{\prod} \mathcal E_t\left(\int_0^\cdot B^{ \S|\D}_s(i)d\bar n^{\mathcal C}_s(i)\right) \\\nonumber
& \times  \exp \left\{ - \Lambda^{\S|\D}_t(\C) - \sum_{j\in\D}\int_0^{T(j) \wedge t} A^{\S-\D}_s(j) \alpha_s(j)ds \right\}.
\end{align}
We have used the following notations:
\begin{align}
\Lambda_t^{\S|\D}(i)&:=\int_0^t\lambda_s^{\S|\D}(i)ds\quad \Lambda_t^{\S}(i):=\Lambda_t^{\S|\emptyset}(i)=\int_0^t\lambda_s^{\S}(i)ds\\\label{lambdaSD}
\lambda^{\S|\D}_t(i) &:=\lambda^{\S-\D}_t (i)+g_t(i) \left(\sum_{j\in\D} \phi^B_s(i, j)\I_{\{T(j)<s\}} \right)\\
\label{BSDformula}
B_t^{\S|\D}(i)&:=\frac{1}{\gamma(i)}\left(\sum_{j\in\S-\D}\phi^B_t(i,j)\I_{\{\tau^B(j)<t\}}+\sum_{j\in\D}\phi^B_t(i,j)\I_{\{T(j)<t\}}\right)
\end{align}

\end{itemize}
\end{prop}
\proof Let us denote: $p^{\S|\D}(t):=\P^\N(\tau (k)>t,\forall k\in\C\;;\;\tau^B(j)\leq t,\forall j\in\D)$. We first show that:
\begin{equation}\label{pc0}
p^{\S|\D}(t)=\bar \E_{\C} \left[L_t^{\S}\I_{\{\tau^B(j)\leq t;\forall j\in\D\}}\right].
\end{equation}
Looking to the formula in (\ref{pc0}), the roadmap is clear: we need to go form the probability $\P^\N$ to $\bar \P_\C$, and from the filtration $\GG^\N$ to $\GG^\S$, as $L^{\S|\D}$ is a $\GG^{\S}$ adapted process. We notice that the Radon-Nikod\'ym density process $D^\S=d\P^\S/d\P^0|_{\G^\N}$, with $\S\subset\N$ is $\GG^\N$ adapted.  

We now introduce some useful $\GG^\S$-adapted processes:
\begin{align*}
E_t^{\mathcal S}&=\underset{i\in \mathcal S}{\prod} \mathcal E_t\left(\int_0^tA_s^{\mathcal S}(i)dm_s(i)\right) \underset{i\in \mathcal N}{\prod}  \mathcal E_t\left(\int_0^tB^{\mathcal S}_s(i)d n_s(i)\right)\\
F_t^{\mathcal S}&=\underset{i\in \mathcal S}{\prod} \mathcal E_t\left(\int_0^tA_s^{\mathcal S}(i)dm_s(i)\right) \underset{i\in \mathcal N}{\prod}  \mathcal E_t\left(\int_0^tB^{\mathcal S}_s(i)d\bar n^\mathcal C_s(i)\right)
\end{align*}
so that
\[
L_t^{\S}=F_t^{\mathcal S}\exp \left (-\Lambda^{\mathcal S}_t( \mathcal C) \right)\
\]
We remark that, indeed, both $E^{\mathcal S}$ and $F^{\mathcal S}$ are  $\GG^{\S}$ adapted. Also, $E^{\mathcal S}$ is a local martingale under $\P^0$, while   $F^{\mathcal S}$ is a  local martingale under $\bar \P_\C$. Furthermore, we have the relation:
\[
F_t^{\mathcal S}= E_t^{\mathcal S}\times \exp\left ( \sum_{k\in \C}\int_0^{t}B^{\S}(k)\beta_s(k)ds\right ).
\] 

It follows that the expression in (b) can be computed as (using the expression (\ref{DTk}) to start with):

\begin{align*}
p^{\S|\D}(t)& =\E^0 \left[D^\N _t\I_{\{\tau(k)>t, \forall k\in\C\}}\I_{\{\tau^B(j)\leq t, \forall j\in\D\}}\right]\\
&=\E^0 \left[D^{\S} _t\I_{\{\tau(k)>t,\forall k\in\C\}}\I_{\{\tau^B(j)\leq t, \forall j\in\D\}}\right]\\
&= \E^0 \left[E^\S_t\exp\left(- \sum_{k\in \C}\int_0^t \alpha_s(k) A^{\S}_s(k) ds  \right)\I_{\{\tau(k)>t,\forall k\in\C\}}\I_{\{\tau^B(j)\leq t, \forall j\in\D\}}\right]\\
&=\E^0 \left[E_t^\S \exp\left(- \sum_{k\in \C}\int_0^t \alpha_s(k) A^{\S}_s(k) ds  \right)\prod_{k\in\C}Z_t(k)\I_{\{\tau^B(j)\leq t, \forall j\in\D\}}\right].
\end{align*}
The last equality is obtained by using the fact that the random variables $e(i),i\in\N$ are independent under $\P^0$ and $\C\cap\S=\emptyset$:
\begin{align*}
\P^0(\tau(k)> t, \forall k\in \C|\G^{\S}_t)&=\E^0\left[\P^0\left(\tau(k)> t,\forall k\in \C|\F_t\vee_{i\in\S}\sigma(e(i))\right)|\G^{\S}_t\right]\\
&=\E^0\left[\P^0\left(\tau(k)> t,\forall k\in \C |\F_t\right)|\G^{\S}_t\right]=\prod_{k\in\C}\P^0\left(\tau(k)> t |\F_t\right).
\end{align*}
To continue, we just need to use the expression  for  $Z(k)$ in (\ref{Z(k)}) and  for $\bar\P_\C$ in Definition \ref{PC}:
\begin{align*}
p^{\S|\D}(t)& =\E^0 \left[E_t^S\exp \left (- \sum_{k\in \C} \int_0^t\left[\lambda_s(k)+\alpha_s(k)A^{\S}_s(k) \right]ds \right )\prod_{k\in\C}\mathcal E_t(\nu(k))\I_{\{\tau^B(j)\leq , \forall j\in\D\}}\right]\\
&=\bar \E_{\C} \left[E^S_t\exp \left (- \sum_{k\in \C} \int_0^t\left[\lambda_s(k)+\alpha_s(k)A^{\S}_s(k) \right]ds \right )\I_{\{\tau^B(j)\leq t, \forall j\in\D\}}\right]\\
&= \bar \E_{\C} \left[ F^\S_t\exp \left (-\sum_{k\in \C}\Lambda^{\S}_t(k) \right) \I_{\{\tau^B(j)\leq t, \forall j\in\D\}}\right]\\
&=\bar \E^{\C} \left[L_t^\S\I_{\{\tau^B(j)\leq t, \forall j\in\D\}}\right],
\end{align*}so that (\ref{pc0}) is proved.
We now prove the particular formulas of our proposition:
\begin{itemize}

\item[(a)] The formula (\ref{pc1}) is obtained from (\ref{pc0}) with $\D=\emptyset$. 

\item[(b)] For the formula (\ref{LSDformula}) a bit more work is needed. On the set $\{\tau^B(j)<t;\forall j\in\D\}$ we have that $L^\S$, which is a $\GG^\S$-adapted process, is equal to some  $\GG^{\S-\D}$-adapted process, that is: 
\begin{equation}\label{LSD}
L_t^\S\I_{\{\tau^B(j)<t;\forall j\in\D\}}=L_t^{\S|\D}\I_{\{\tau^B(j)<t;\forall j\in\D\}},
\end{equation}with $L^{\S|\D}$ being   $\GG^{\S-\D}$-adapted.
We proceed to identify the process $L^{\S|\D}$ (basically this consists in, for all $j\in\D$, replacing  $\tau(j)$ with $T(j)$ as they are equal on the set $\{\tau^B(j)<\infty\}$). We need to show it corresponds to the expression in (\ref{LSDformula}). 

First, we notice that:
\begin{align*}
&\underset{i\in \S}{\prod} \mathcal E_t\left(\int_0^\cdot A_s^{ \S}(i)dm_s(i)\right)\I_{\{\tau^B(j)<t;\forall j\in\D\}}= \\
&=\exp \left (-\sum_{j\in\D}\int_0^{t\wedge T(j)}A^{\S-\D}_s(j)\alpha (j)ds\right)\underset{i\in \S-\D}{\prod} \mathcal E_t\left(\int_0^\cdot A_s^{ \S-\D}(i)dm_s(i)\right)\I_{\{\tau^B(j)<t;\forall j\in\D\}}
\end{align*}
Second, for $k\in \S-\D$, we have:
\begin{align*}
B^\S_t(k)\I_{\{\tau^B(j)<t;\forall j\in\D\}}&= B^{\S|\D}_t(k)\I_{\{\tau^B(j)<t;\forall j\in\D\}}\\
\lambda^S_t(k)\I_{\{\tau^B(j)<t;\forall j\in\D\}}&=\lambda^{\S|\D}_t(k)\I_{\{\tau^B(j)<t;\forall j\in\D\}}.
\end{align*}
with $B^{\S|\D}$ and $\lambda^{\S|\D}$ being  $\GG^{\S-\D}$-adapted;  $B^{\S|\D}$ is given in (\ref{BSDformula}), and $\lambda^{\S|\D}$ in (\ref{lambdaSD}). We therefore identify $L^{\S|\D}$ as the one in (\ref{LSDformula}).

Using the relation (\ref{LSD}) in the formula (\ref{pc0}) and then the fact that $L^{\S|\D}$ is $\GG^{\S-\D}$ adapted, we obtain:
\begin{align*}
p^{\S|\D}(t)&=\bar \E_{\C} \left[L_t^{\S|\D}\I_{\{\tau^B(j)\leq t, \forall j\in\D\}}\right]\\
&=\bar \E_{\C} \left[L_t^{\S|\D}\bar \P_\C\left (\tau^B(j)\leq t, \forall j\in\D|\G^{\S-\D}_t\right )\right].
\end{align*}
Furthermore, 
\begin{align*}
\bar \P_\C&\left (\tau^B(j)\leq t, \forall j\in\D|\G^{\S-\D}_t\right )= \P^0\left (\tau^B(j)\leq t, \forall j\in\D|\G^{\S-\D}_t\right )\\
&=\E^0\left[\P^0\left (\tau^B(j)\leq t, \forall j\in\D|\F_t\vee_{k\in\S-\D}\sigma(e(k))\right )|\G^{\S-\D}_t\right]\\
&=\E^0\left[\P^0\left (\tau^B(j)\leq t, \forall j\in\D|\F_t\right )|\G^{\S-\D}_t\right]=\P^0\left (\tau^B(j)\leq t, \forall j\in\D|\F_t\right )\\
&=\prod_{j\in\D}\P^0\left (\tau^B(j)\leq t|\F_t\right )=\prod_{j\in\D}p_t(j)\I_{\{T(j)\leq t\}}.
\end{align*}
Above, we have used for obtaining the first equality, the fact that the Radon-Nikod\'ym density process $\frac{d\bar \P_\C}{d\P^{0}} |_{\G_t^\N}$ is $\FF$ adapted, hence also $\GG^{\S-\D}$ adapted and for the second equality, the fact that $\G_t^{\S-\D}\subset \F_t\vee_{k\in\S-\D}\sigma(e(k))$, so that (\ref{pc2}) is proved.
\end{itemize}
\finproof

%

\bigskip

The dynamics of the processes $ \ell^\S$  and $ \ell^{\S|\D}$ will be obtained from intermediary quantities, falling basically into two categories:
\begin{prop}\label{BB} The following hold, for $j\in\S-\D$:
\begin{align}\label{BB1}
(a)\;&\bar\E_\C[L_t^{\S|\D} \I_{\{\tau^B(j)<t\}}|\F_t]= \ell_t^{\S|\D\cup j}p_t(j)\I_{\{T(j)<t\}},\\\label{BB2}
(b)\;&\bar \E_\C[L^{\S|\D}_t\I_{\{\tau^A(j)<t\}}|\F_t]= \ell^{S|\D}_t- \ell^{\S-j|\D}_t-  \ell^{\S|\D\cup j}_tp_t(j)\I_{\{T(j)<t\}}.
\end{align}
We recall that  $ \ell^{\S|\D}$ is the $(\FF,\bar \P_\C)$-optional projection of the process $L^{\S|\D}$. Consequently, $ \ell^{\S-j|\D}$ is the $(\FF,\bar \P_{\C\cup \{j\}})$-optional projection of the process $L^{\S-j|\D}$.
\end{prop}

\begin{proof} 

 We prove (\ref{BB1}). It can be directly checked that for $j\notin\D$,  $L_t^{\S|\D} \I_{\{\tau^B(j)<t\}}=L_t^{\S|\D\cup j} \I_{\{\tau^B(j)<t\}}$. Therefore (using same path as in the proof of Proposition \ref{MainProp} (b)):
\begin{align*}
\bar\E_\C[L_t^{\S|\D} \I_{\{\tau^B(j)<t\}}|\F_t]
&=\bar \E_\C[L_t^{\S|\D\cup j} \I_{\{\tau^B(j)<t\}}|\F_t]\\
&=\bar \E_\C\left [L_t^{\S|\D\cup j}\bar \P_C\left( \tau^B(j)<t |\G^{\S-\D-j}\right)|\F_t\right ]\\
&=\bar \E_\C\left [L_t^{\S|\D\cup j} p_t(j)\I_{\{T(j)<t\}}|\F_t\right ]=\bar \E_\C\left [L_t^{\S|\D\cup j}|\F_t\right ] p_t(j)\I_{\{T(j)<t\}}\\
&= \ell_t^{\S|\D\cup j}p_t(j)\I_{\{T(j)<t\}}.
\end{align*}

On the other hand, we have 
\[
\{\tau^A(j)<t\}=\left(\{\tau(j)\geq t\}\cup\{\tau^B(j)<t\}\right)^c.
\]so that (again with $\D\subset\S$, $j\in\S-\D$):

\begin{align*}
\bar \E_\C[L^{\S|\D}_t\I_{\{\tau^A(j)<t\}}|\F_t]&= \ell^{\S|\D}_t-\bar \E_\C[L^{\S|\D}_t \I_{\{\tau(j)\geq t\}}|\F_t] - \bar \E_\C[L^{\S|\D}_t \I_{\{\tau^B(j)<t\}}|\F_t]\\
&= \ell^{S|\D}_t- \ell^{\S-j|\D}_t-  \ell^{\S|\D\cup j}_tp_t(j)\I_{\{T(j)<t\}}.
\end{align*}
\end{proof}

\subsection{Proof of Theorem \ref{MainThm}}  We denote by  $ \ell^{\S|\D}$ (resp. $ \ell^{\S}$) is  the $(\FF,\bar \P_\C)$ optional projection of $L^{\S|\D}$ (resp. $L^{\S}$). Then the expressions (\ref{ellSD}) and (\ref{ellS}) are a consequence of Proposition \ref{MainProp}:
\[
p^{\S|\D}(t)=\bar \E_{\C} \left[L_t^{\S|\D}\prod_{j\in\D}p_t(j)\I_{\{T(j)\leq t\}}\right]=\bar \E_{\C} \left[\ell_t^{\S|\D}\prod_{j\in\D}p_t(j)\I_{\{T(j)\leq t\}}\right]
\]
It remains to determine the dynamics of  $ \ell^{\S|\D}$ and  $ \ell^{\S}$.

We notice that the stated dynamics of $ \ell^\S$ in (\ref{LS}) coincide with those of $ \ell^{\S|\emptyset}$, derived from  (\ref{L}), when taking $\D=\emptyset$. Therefore, it is only needed to prove that for a general $\D\subset \S$ the dynamics of $ \ell^{\S|\D}$ in (\ref{L}) are correct. 

To do so, we start from the SDE corresponding to $L^{\S|\D}$.  From (\ref{LSDformula}), we have:
\begin{align*}
L^{\S|\D}_t=& \;1+\sum_{i\in \S-\D}\int_0^t L^{\S|\D}_{s-}A^{\S-\D}_s(i)dm_s(i)+ \sum_{i\in \N}\int_0^tL^{\S|\D}_{s-} B^{\S|\D}_s(i)d\bar n^{\C}_s(i)\\ &-\sum_{i\in\C}\int_0^tL^{\S|\D}_{s-}\lambda_s^{\S|\D}(i)ds-\sum_{i\in\D}\int_0^{t\wedge T(i)}  L^{\S|\D}_{s-}\alpha_s(i)A^{\S-\D}_s(i)ds
\end{align*}
The process $ \ell^{\S|\D}$ being the  $(\FF,\bar \P_\C)$ optional projections of   $L^{\S|\D}$,  for finding  $\ell^{\S|\D}$, we  compute the $(\FF,\bar \P_\C)$ optional projections of each term on the right hand side of the above expression. It is important to emphasise that the filtration $\FF$ is immersed in the filtration $\GG^{\S-\D}$ under the measure $\bar\P_\C$ (see Lemma \ref{filtrations} (c)). Therefore, we can use the classical projection formulas summarised in the Appendix (Proposition \ref{bremyorproj} and Lemma \ref{projLemma}).

First, we have:
\begin{lem} For all $i\in\S-\D$ and $t\geq 0$: 
\begin{align*}
\bar\E_{\C} \left [\int_0^t L^{\S|\D}_{s-}A^{\S-\D}_s(i)dm_s(i)|\F_t\right ]&= 0.
\end{align*}
\end{lem}
\proof
We fix some $i\in\S-\D$. We notice that $m(i)$ is a $(\GG^{\S-\D},\bar\P_\C)$ martingale and the process $H(i):=L^{\S|\D}_{\cdot -}A^{\S-\D}(i)$ is $\GG^{\S-\D}$-predictable. As an application of the Lemma \ref{projLemma}, it follows that the $(\GG^{\S-\D-i},\bar \P_\C)$-optional projection of $\int H(i)dm(i)$ is null (therefore also the $(\FF,\bar \P_\C)$-optional projection). Indeed, taking  $\rho=\tau^A(i)$, $\HH:=\GG^{\S-\D-i}$,  we observe that  the conditions for applying  Lemma \ref{projLemma} are fulfilled:  the filtrations $\GG^{\S-\D-i}$ and $\GG^{\S-\D}$ are immersed under $\bar\P_\C$ (Lemma \ref{filtrations} (c)), the process $H(i)$ is here bounded and $\tau^A(i)$ avoids all $\GG^{\S-\D-i}$ stopping times.
\finproof

Secondly:
\begin{lem} For all $i\in\N$ and $t\geq 0$,
\begin{align*}
\bar\E_\C &\left[ \int_0^t L^{\S|\D}_{s^-}  B^{\S|\D}_s(i)d\bar  n^\C_s(i)  |\F_t\right ] = \int_0^t \bar\E_\C\left[L^{\S|\D}_{s^-}B^{\S|\D}_s(i)|\F_s\right ]\ d\bar n^\C_s(i) 
\end{align*}
\end{lem}

\proof
It is a direct application of Proposition \ref{bremyorproj} (i), with $\HH:=\FF$, $\GG:=\GG^{\S-\D}$,  $M:=\bar n^\C(i)$ and $G:=L^{\S|\D}_{s^-}  B^{\S|\D}_s(i)$.
\finproof

It follows form the last two lemmas that the $ \ell^{\S|\D}$ writes:
\begin{align}\nonumber
 \ell^{\S|\D}_t=& \;1+\sum_{i\in \N} \int_0^t\bar\E_\C\left[L^{\S|\D}_{s^-}B^{\S|\D}_s(i)|\F_s\right ]\ d\bar n^\C_s(i)- \sum_{i\in\C} \int_0^t \bar\E_\C\left[L^{\S|\D}_{s^-} \lambda_s^{\S|\D}(i)|\F_s\right ] ds\\\label{prelimLhat}
& -\sum_{i\in\D} \int_0^{t\wedge T(i)} \alpha_s(i)\bar\E_\C\left[L^{\S|\D}_{s^-}A^{\S-\D}_s(i)|\F_s\right ] ds 
\end{align}

The expression above contains some conditional expectations that we now compute explicitly, with the help of Proposition \ref{BB}.

 For $i\in\D$ and with $\alpha_t>0$:

\begin{align*}
 \bar\E_{\C}\left[ L^{\S|\D}_{t^-}A^{\S-\D}_t(i)|\F_t\right] &= \bar\E_{\C}\left[ L^{\S|\D}_{t-}\frac{1}{\alpha_t(i)}\sum_{j\in \S-\D}\phi^A_{t}(i,j)\I_{\{\tau^A(j)< t\}}|\F_t\right]\\
 &=\frac{1}{\alpha_t(i)} \sum_{j\in \S-\D}\phi^A_{t}(i,j)\bar\E_\C\left[  L^{\S|\D}_{t^-}\I_{\{\tau^A(j)< t\}}|\F_t\right]\\
 &=\sum_{j\in \S-\D}  \left( \ell^{\S|\D}_{t^-}- \ell^{\S-j|\D}_{t^-} -  \ell^{\S|\D\cup j}_{t^-} p_t(j)\I_{\{T(j)<t\}}\right)\frac{ \phi^A_{t}(i,j)}{\alpha_t(i)}
\end{align*}
(we used Proposition \ref{BB} in the last step).

On the other hand, for $i\in\N$ and with $\gamma_t>0$:
\begin{align*}
 \bar\E_{\C}\left[ L^{\S|\D}_{t^-}B^{\S|\D}_t(i)|\F_t\right] &=  \frac{1}{\gamma_t(i)} \bar\E_{\C}\left[ L^{\S|\D}_{t^-}\left(\sum_{j\in \S-\D}\phi^B_{t}(i,j)\I_{\{\tau^B(j)< t\}}+\sum_{j\in\D}\phi^B_{t}(i,j)\I_{\{T(j)< t\}}\right)|\F_t\right]\\
 &=\sum_{j\in \S}  \left(\I_{\{j\in\D\}} \ell^{\S|\D}_{t^-}+\I_{\{j\in\S-\D\}} \ell^{\S|D\cup j}_{t^-}p_t(j)\right)\frac{\phi^B_{t}(i,j) }{\gamma_t(i)}\I_{\{T(j)<t\}}.\\
\end{align*}

Finally, for $i\in \C$, and using the two above computed quantities:
\begin{align*}
&\bar\E_\C \left[L^{\S|\D}_{t^-}\lambda^{\S|\D}_t(i)|\F_t\right]=\\
&=\bar\E_\C \left[L^{\S|\D}_{t^-}\left\{\lambda_t(i)+ \alpha_t(i)A^{\S-\D}_t(i)+\beta_t(i)B^{\S|\D}_t(i)\right\}|\F_t\right]\\
&= \ell^{\S|\D}_{t^-}\lambda_t(i)+ \alpha_t(i)\bar\E_\C \left[ L^{\S|\D}_{t^-}A^{\S-\D}_t(i)|\F_t\right]+ \beta_t(i)\bar\E_\C \left[ L^{\S|\D}_{t^-}B^{\S-\D}_t(i)|\F_t\right]\\
&= \ell^{\S|\D}_{s^-}\left(\lambda_t(i)+ \phi^A_{s}(i,\S-\D)+g_t(i)\sum_{j\in\D}\phi^B(i,j)\I_{\{T(j)<t\}}\right) \\
&\quad  - \sum_{j\in\S-\D}  \left[ \ell^{S-j|\D}_{t^-} \phi^A_t(i,j)  +  \ell^{\S|\D\cup j}_{t^-} \left [\phi^A_t(i,j)-g_t(i)\phi^B_{t}(i,j)\right] p_t(j)\I_{\{T(j)<t\}} \right].
\end{align*}

We now replace  the conditional expectations  in (\ref{prelimLhat}) with the terms computed above; we obtain the following:

\begin{align*}
d \ell^{\S|\D}_t&= \sum_{i\in \N} \sum_{j\in \S}\left(\I_{\{j\in \D\}}  \ell^{\S|\D}_{t^-}+\I_{\{j\in\S-\D\}}  \ell^{\S|\D\cup j}_{t^-} p_t(j)\right)\frac{\phi^B_{t}(i,j)}{\gamma_t(i)}\I_{\{T(j)<t\}}d\bar n^{\C}_t(i)\\
&-\sum_{i\in\C}\Big\{\ell^{\S|\D}_{t^-}\left(\lambda_t(i)+ \phi^A_{t}(i,\S-\D)+g_t(i)\sum_{j\in\D}\phi^B_t(i,j)\I_{\{T(j)<t\}}\right)  \\
&\quad\quad \quad - \sum_{j\in\S-\D}  \left[ \ell^{S-j|\D}_{t^-} \phi^A_t(i,j)  +  \ell^{\S|\D\cup j}_{t^-} \left [\phi^A_t(i,j)-g_t(i)\phi^B_{t}(i,j)\right] p_t(j)\I_{\{T(j)<t\}} \right]\Big \}dt\\
&-\sum_{i\in\D}\I_{\{T(i)>t\}}\sum_{j\in \S-\D}  \phi^A_{t}(i,j) \left( \ell^{\S|\D}_{t-}- \ell^{\S-j|\D}_{t-} -  \ell^{\S|\D\cup j}_{t-} p_t(j)\I_{\{T(j)<t\}}\right)dt
\end{align*}


We rearrange terms and  use:
\[
 \bar n^\C_t(i)= 
 \begin{cases}
 n_t(i)+\int_0^{t\wedge T(i)} \gamma_s(i) g_s(i)ds & i \in \C \\
 n_t(i) & i \in \S
 \end{cases}
 \] which follows from (\ref{nc}) and the remark thereafter.
Also, we denote 
\[
\psi_t^A(i,j):=
\begin{cases} 
\phi_t^A(i,j) & i \in \C \\
\phi_t^A(i,j)\I_{\{T(i)>t\}} & i \in \D
\end{cases}
\]

We obtain the dynamics of $ \ell^{\S|\D}$:
\begin{align*}
d \ell^{\S|\D}_t=
&\;  \ell^{\S|\D}_{t^-}\left\{\left[ \lambda_t(\C)+ \psi^A_{t}(\C\cup\D,\S-\D)\right]dt+\sum_{j\in\D} \I_{\{T(j)<t\}}\sum_{i\in \N} \left(\frac{\phi^B_{t}(i,j)}{\gamma_t(i)}\right)d n_t(i)\right\}\\
&+  \sum_{j\in \S-\D} \ell^{\S|\D\cup j}_t p_t(j)\I_{\{T(j)<t\}}\left\{\psi_t^A(\C\cup\D,j)dt +\sum_{i\in \N}  \left(  \frac{\phi^B_{t}(i,j)}{\gamma_t(i)}\right) dn_t(i)\right\}\\
&+  \sum_{j\in \S-\D} \ell^{\S-j|\D}_t\psi_t^A(\C\cup\D,j)dt.
\end{align*}
This is nothing but another form of (\ref{L}), so that the result is proved.

\appendix
\renewcommand{\thesubsection}{\Alph{subsection}}
\section{\texorpdfstring{$\F$-Conditional Markov chains}{F-Conditional Markov chains} }\label{AppendixDef}
A heuristic description of a process that is a Markov chain conditionally to a sigma-field $\F$, is to consider a two step randomisation procedure. In the first step, one draws at random the trajectory of a ``driving process'', i.e., a stochastic process that  is interpreted as the stochastic environment and that generates the filtration $\FF=(\F_t)_{t\geq 0}$ with $\F:=\F_\infty$. Once a whole trajectory is selected, one generates a Markov chain  with transition rates being functions of the current state of the driving process. A popular example of conditional Markov chain is the doubly stochastic Poisson process (see \cite{Sny75}). However, in the framework of default contagion, we consider a process $\Y$ with state space $I=\{0,1\}^n$, hence exclude the doubly stochastic Poisson processes.

Now,  given a probability space $(\Omega, \F,\P)$, let us formalise the definition of a $\F$-conditional Markov chain with state space $I$. We denote $2^I$ the set of all subsets of $I$; $\F$ is  a sigma-field.

An $\F$-conditional transition probability $\pi$ is a map from $I\times2^I\times\Omega $ into $[0,1]$ such that
\begin{itemize}
\item[(i)] for all $\mathbf x\in I$ and $\omega\in\Omega$, the map $A\to \pi(\mathbf x,A,\omega)$ is a probability measure on $2^I$.
\item[(ii)] for all $A\in 2^I$ and $\omega\in\Omega$, the map $\mathbf x\to \pi(\mathbf x,A,\omega)$ is  $2^I$ measurable.
\item[(iii)] for all $\mathbf x\in I$ and $A\in 2^I$, the map $\omega \to \pi(\mathbf x,A,\omega)$ is  $\F$ measurable.
\end{itemize}

We say that $\Y$ is a $\F$-conditional Markov chain, if for any fixed $s$ and $t$ with $0\leq t\leq s$, there is an $\F$-conditional transition probability:
\begin{align*}
p_{t,s}:\;I\times2^I\times\Omega & \to [0,1]\\
(\mathbf x,A,\omega)&\to p_{t,s}(\mathbf x,A)(\omega)
\end{align*}
 such that the following holds:
\[
\P(\Y_s\in A|\F\vee \sigma(\Y_u,u\leq t))(\omega)=p_{t,s}(\Y_t,A)(\omega).
\]
In particular, $p_{t,s}(\Y_t,A)$ is $\F$-measurable. Hence for fixed $\omega \in \Omega$, $p_{t,s}(\mathbf x,\mathbf y)$ is the transition function of a time inhomogeneous Markov chain from state $\mathbf x$ to state $\mathbf y$. By being $\F$-measurable, the transition functions are random, hence capturing the dependence of the process $\Y$ on  the stochastic environment.

In general, the distribution of the process $\Y$ conditionally on $\F$, can be synthesised by the so-called instantaneous transition rates, or intensities of transition. When these exist, they are defined as follows. For $\mathbf x\in I$ and $\mathbf y\in I$:
\[
q_t(\mathbf x,\mathbf y)=\lim_{\epsilon \to 0}\frac{p_{t,t+\epsilon}( \mathbf  x,\mathbf  y)}{\epsilon}.
\] 
As the transition functions are random variables, the convergence considered is the $\P$-almost sure convergence.

\section{Basic facts in enlargement of filtrations}\label{Appendix}

Here we summarise the results from the theory of enlargements of a filtration that were useful in this paper. 

We assume we are given a filtered probability space $\left(\Omega,\F,\HH=(\mathcal H_t)_{t\geq 0},\P\right)$  satisfying the usual assumptions. 

\subsection{Progressive enlargement}

\begin{defn}
A random time $\rho$ is a nonnegative random variable $\rho:\left(\Omega,\mathcal{F}\right)\rightarrow[0,\infty]$.
\end{defn}

The  Az\'ema supermartingale associated to $\rho$ and relative to $(\HH,\P)$ is the $\HH$ supermartingale
\begin{equation}
Z_{t}^{\rho }=\mathbb{P}( \rho >t\mid \mathcal{H}_{t})
\label{surmart}
\end{equation}
chosen to be c\`{a}dl\`{a}g, associated with $\rho $\ by Az\'{e}ma (Az\'ema \cite{azema}).  We note that the supermartingale $\left(Z_{t}^{\rho}\right)$ is the $\HH$-optional projection of $\mathbf{1}_{[0,\rho[}$.  We also introduce the $\HH$ dual optional and dual predictable projections of the process $\I_{\left\{ \rho \leq t\right\} }$, denoted respectively by $A_{t}^{\rho }$ and $a_{t}^{\rho}$.
Then,
\[
Z_{t}^{\rho }=\E^\P[A^\rho_\infty|\mathcal H_t]-A_{t}^{\rho}.
\]
while  the Doob-Meyer decomposition of (\ref{surmart}) writes:
\begin{equation}\label{doobdecZ}
Z_{t}^{\rho }=m_{t}^{\rho }-a_{t}^{\rho}.
\end{equation}

We enlarge the initial filtration $\HH$ with the process $(\rho\wedge t) _{t\geq 0}$, so that the new enlarged filtration $\mathbb{H}^{\rho}$\  is the smallest filtration (satisfying the usual assumptions) containing $\HH$\ and making $\rho$ a stopping time, that is:
\[
\mathcal H^\rho_t=\mathcal{K}_{t+},\text{ where }\mathcal{K}_{t}=\mathcal{H}_{t}\vee\sigma(\rho\wedge t).
\]

Now we recall a theorem which is useful in constructing the $(\HH^\rho,\P)$ compensator process of $\rho$. 
\begin{thm}[Jeulin-Yor \cite{yorjeulin}]\label{calccomp}
Let $H$ be a bounded $\HH^\rho$ predictable process. Then
$$H_{\rho}\mathbf{1}_{\{\rho\leq t\}}-\int_{0}^{t\wedge\rho}\dfrac{H_{s}}{Z_{s-}^{\rho}}da_{s}^{\rho}$$is
a $\HH^\rho$ martingale.
\end{thm}

When one assumes that the random time $\rho$ avoids $\HH$ stopping times, then:
\begin{lem}[Jeulin-Yor \cite{yorjeulin}, Jeulin \cite{jeulin}]\label{lem:evitement}
If $\rho $\ avoids  $\HH$ stopping times,  then $A^{\rho }=a ^{\rho }$  and
$A^\rho$ is continuous. Therefore, the compensator of the process $\I_{\{\rho\leq t\}}$ is continuous.
\end{lem}

\subsection{Immersion of filtrations}
Given two filtrations $\HH$ and $\GG$, with $\mathcal H_t\subset\G_t$, for all $t\geq 0$, the following assumption is often encountered in the literature:

 The filtration $\HH$ is immersed in $\GG$ (also called (H)-hypothesis): every $\HH$ martingale is a $\GG$ martingale.
 
We write $\HH \overset{\P}{\hookrightarrow} \GG$ for $\HH$ is immersed in $\GG$ under the probability measure $\P$.

We now recall several useful
equivalent characterizations of the immersion property in
the next theorem
\begin{thm}[Dellacherie-Meyer \cite{delmeyfil} and Br\'{e}maud-Yor \cite{bremaudyor}]\label{lecasdregeneral} The following assertions are equivalent:

\begin{enumerate}
\item $\HH \overset{\P}{\hookrightarrow} \GG$ ;

\item For all bounded $\mathcal H_{\infty }$-measurable random variables $H$\ and all bounded $\mathcal{G}_{t}$-measurable random variables $G_{t}$, we have
\begin{equation*}
\E^\P\left[HG_{t}| \mathcal{H}_{t}\right] =\E^\P\left[ H| \mathcal{H}_{t}\right] \E^\P\left[ G_{t}| \mathcal{H}_{t}\right] .
\end{equation*}

\item For all bounded $\mathcal{H}_{\infty }$ measurable random variables $H$,
\begin{equation*}
\E^\P\left[ H\mid \mathcal{G}_{t}\right] =\E^\P\left[ H \mid \mathcal{H}_{t}\right] .
\end{equation*}

\end{enumerate}
\end{thm}

The immersion property is preserved only by certain changes of the probability measure. One such example is the following:
\begin{prop}[Jeulin-Yor \cite{yorjeulinens}]\label{proplL} We assume that $\HH \overset{\P}{\hookrightarrow} \GG$.
Let $\Q$ be a probability measure which is equivalent to $\P$ on
$\G_\infty$. If
$d\Q/d\P$ is $\mathcal{H}_{\infty}$-measurable, then $\HH
\overset{\Q}{\hookrightarrow} \GG$.
\end{prop}


One advantage of the immersion property is that optional projections of some $\GG$  adapted processes can be computed easily. We recall the projection formulas that were useful in the derivation of our main result. 
%

\begin{prop}[Br\'emaud-Yor \cite{bremaudyor}]\label{bremyorproj} Suppose that $\HH \overset{\P}{\hookrightarrow}  \GG$.
\begin{itemize}
\item[(i)] Let $M$ be an $\HH$ local martingale and $G$ be a $\GG$ adapted and bounded process. Then the $\HH$ optional projection of the process $\left(\int GdM\right)$ is given by  $\int \;^oGdM$, where $^oG$ is the $\HH$ optional projection of $G$.
\item[(ii)] If $M$ is a $\GG$ square integrable martingale and $H$ an $\HH$ adapted and bounded process. Then the $\HH$ optional projection of the process $\left(\int HdM\right)$ is given by  $\int Hd\;^oM$, where $^oM$ is the $\HH$ optional projection of $M$.
\end{itemize}
\end{prop}

In the framework and with the notations of the previous subsection, we have:
 \begin{lem}[Coculescu et al. \cite{cocjeanik09}]\label{projLemma} Assume that $\rho$ avoids all $\HH$ stopping times and  $\HH \overset{\P}{\hookrightarrow} \HH^\rho$ holds. Let  $H$ be a   $\GG$-predictable process and let $N_t=\I_{\{\rho\leq t\}}-\Gamma_{t\wedge\rho}$ be a  $\GG$ martingale. If $\E^\P[|H_\rho|]<\infty$, then the $\HH$ optional projection of the process $\left(\int HdN\right)$ is null.
\end{lem}

\section{On the link with the Markovian approach}\label{AppendixC}

We consider the model that we have introduced of Section \ref{sec3}. 
A Markovian model similar to the one presented in Section \ref{sec2} can be obtained by imposing the following:

\textbf{Assumption:}  All the stopping times $\tau(i), i\in\N$ avoid the $\FF$ stopping times.

This is equivalent to take  $p_t(i)=\P^0(\tau(i)=T(i)|\F_t)=0$ for all $t\geq 0$ and for all  $i\in\N$, that is: $\tau^B(i)=\infty\;a.s.$ and $\tau(i)=\tau^A(i)$ a.s. It follows that, under the measure $\P^\N$ there is now only direct contagion, as $B^\N(i)\equiv 0$; the default times do not have any impact on the environment.

We now assume this is the case. Then, under the reference probability $\P^\N$ and conditionally on $\F_\infty$, $ \Y$ is a $n$-dimensional, time inhomogeneous  Markov chain with state space $I:=\{0,1\}^n$, conditioned to start at $Y_0=0$\footnote{Indeed, in Section \ref{sec3}, we have set $\mathbf \Gamma_0=(0,..,0)$ that is $\mathbf Y_0=(0,..,0)$ $\P^\N$-a.s.; the proof of the Markov chain property is trivial.}.

Let us show how the SDE (\ref{LSMarkov}) can be obtained from the Kolmogorov forward equations.   In order to have simple notations (and in particular avoid to introduce an ordering of the states of $Y$) we denote, for $\D\subset \N$:
\begin{align*}
p_t^\D:&=\P^\N(Y_t(i)=0,\forall i\in\N-\D\;;\;  Y_t(j)=1,\forall j\in\D\;|\;\F_\infty),
\end{align*}i.e., the probability that $\D$ is the set of defaulted debtors at time $t$, conditionally on $\F_\infty$. We do not need a more complex notation, as we shall here analyse only probabilities conditional on  $\F_\infty$ and $\Y_0=0$.
Also, in the same spirit to simplify notation, we denote:
\[
c_t^\D(k):= q_t(\mathbf x, \mathbf x^k)\text{ with } \D=\{i; x(i)=1\},
\]
hence $c_t^\D(k)$ is the transition rate at time $t$ of debtor $k$, given that $\D$ is the set of defaulted entities at time $t$.
If  the set of defaulted entities is $\D$, after one transition of $\Y$ the set of defaulted entities becomes necessarily $\D\cup\{k\}$, for some $k\in\N-\D$. The corresponding instantaneous transition rate is:
\begin{equation}\label{trans}
c_t^{\D}(k)=\lambda_t(k)+\phi^A_t(k,\D).
\end{equation}
The Kolmogorov forward equations write (using the fact that only one default can occur at a time and that default is  an absorbing state), for all $\D\subset \N$:
\begin{equation}\label{KFE}
\frac{d}{d t}p^\D_t=-p_t^\D c^{\D} _t(\N-\D)+\sum_{k\in \D}p^{\D-k}_tc_t^{\D-k}(k)
\end{equation}
where we use the notation $c^{\D} _t(\N-\D)=\sum_{k\in \N-\D}c_t^{\D}(k)$.
\begin{lem}
We fix a set $\C\subset \N$ and denote $\S:=\N-\C$. Also, we denote:
\[
 \ell^\S_t:=\P(Y_t(i)=0,\forall i\in\C|\F_\infty)
\]Then, $ \ell^\S$ satisfies:
\begin{align}\label{LSMarkov2}
d \ell^\S_t=&  - \ell^\S_t c^\S_t(\C)dt+\sum_{j\in\S}    \ell^{\S-j}_{t}\phi^A_t(\C,j) dt.
\end{align}
\end{lem}
\begin{rem}
The dynamics (\ref{LSMarkov2}) is the same as  (\ref{LSMarkov}) because $c^\S_t(\C)=\lambda_t(\C)+\phi^A_{s}(\C,\S)$, as defined in  (\ref{trans}). 
\end{rem}
\proof

We notice that $ \ell^\S_t:=\sum_{\D:\D\subset \S}p^\D_t$, and therefore, using (\ref{KFE}):
\[
\frac{d}{d t} \ell^\S_t= -\sum_{\D:\D\subset \S}p^\D_t c^{\D} _t(\N-\D)+\sum_{\D:\D\subset \S} \sum_{j\in \D}p^{\D-j}_tc_t^{\D-j}(j).
\]
We fix a set $\D\subset \S$. We denote $s=\mathbf{card}(\S)$ and $n=\mathbf{card}(\D)$. The probability $p^\D_t$ appears in exactly $s-n+1$ terms in the sum above, namely: $-p^\D_t c^{\D} _t(\N-\D)$ and, for all $j\in \S-\D$, the term $p^{\D}_tc_t^{\D}(j)$. Therefore, the expression can be written as:

\begin{align*}
\frac{d}{d t} \ell^\S_t&= \sum_{\D:\D\subset \S}\left (-p^\D_t c^{\D} _t(\N-\D)+\sum_{j\in \S-\D}p^{\D}_tc_t^{\D}(j)\right )\\
&=- \sum_{\D:\D\subset \S}p^\D_t c^{\S} _t(\N-\S)+ \sum_{\D:\D\subset \S}p^\D_t \phi^A_t(\N-\S, \S-\D)\\
&= -l^\S_t c^{\S} _t(\N-\S)+ \sum_{\D:\D\subset \S}\sum_{j\in \S-\D} p^\D_t \phi^A_t(\N-\S, j)\\
&= -l^\S_t c^{\S} _t(\N-\S)+ \sum_{j\in \S}  \sum_{\D\subset \S-\{j\}}p^\D_t \phi^A_t(\N-\S, j)\\
&= -l^\S_t c^{\S} _t(\N-\S)+ \sum_{j\in \S} \left(\sum_{\D\subset \S-\{j\}}p^\D_t\right)  \phi^A_t(\N-\S, j) \\
&= -l^\S_t c^{\S} _t(\N-\S)+ \sum_{j\in \S} l^{\S-j} \phi^A_t(\N-\S, j).
\end{align*}\finproof

\renewcommand{\refname}{References}

\end{document}